\documentclass[10pt,a4paper]{article}

\usepackage{fullpage}
\usepackage{graphicx}
\usepackage{comment}

\newtheorem{theorem}{Theorem}[section]
\newtheorem{lemma}[theorem]{Lemma}

\usepackage{algorithm}
\usepackage[noend]{algpseudocode}

\usepackage[hang]{subfigure}

\newcommand{\remove}[1]{}

\newenvironment{proof}[1][IEEEproof]{\begin{trivlist}
\item[\hskip \labelsep {\bfseries #1}]}{\end{trivlist}}
\newenvironment{definition}[1][Definition]{\begin{trivlist}
\item[\hskip \labelsep {\bfseries #1}]}{\end{trivlist}}

\newcommand{\qed}{\nobreak \ifvmode \relax \else
      \ifdim\lastskip<1.5em \hskip-\lastskip
      \hskip1.5em plus0em minus0.5em \fi \nobreak
      \vrule height0.75em width0.5em depth0.25em\fi}

\begin{document}

\title{TinyLFU: A Highly Efficient Cache Admission Policy}
\date{}
\author{
	Gil Einziger\thanks{Much of the work was done while this author was with the Technion}\\
    Dipartimento di Elettronica\\
    Politecnico di Torino\\
    Italy\\
    \texttt{gil.einziger@polito.it}
\and
    Roy Friedman\\
    Computer Science Department\\
    Technion\\
    Haifa 32000, Israel\\
    \texttt{roy@cs.technion.ac.il}
\and
    Ben Manes\\
    Independent\\
    \texttt{Ben.Manes@gmail.com}
}
\maketitle

\begin{abstract}
This paper proposes to use a \emph{frequency based cache admission policy} in order to boost the effectiveness of caches subject to skewed access distributions.
Given a newly accessed item and an eviction candidate from the cache, our scheme decides, based on the recent access history, whether it is worth admitting the new item into the cache at the expense of the eviction candidate.

Realizing this concept is enabled through a novel approximate LFU structure called \emph{TinyLFU}, which maintains an approximate representation of the access frequency of a \emph{large sample} of recently accessed items.
TinyLFU is very compact and light-weight as it builds upon Bloom filter theory.

We study the properties of TinyLFU through simulations of both synthetic workloads as well as multiple real traces from several sources.
These simulations demonstrate the performance boost obtained by enhancing various replacement policies with the TinyLFU eviction policy.
Also, a new combined replacement and eviction policy scheme nicknamed \emph{W-TinyLFU} is presented.
W-TinyLFU is demonstrated to obtain equal or better hit-ratios than other state of the art replacement policies on these traces.
It is the only scheme to obtain such good results on all traces.
\end{abstract}

\thispagestyle{empty}

\newpage

\section{Introduction}

\emph{Caching} is one of the most basic and effective methods in computer science for boosting system's performance in a multitude of domains.
It is obtained by keeping a small percentage of the data items in a memory that is faster and/or closer to the application in settings where the
entire data domain
does not fit into this fast nearby memory.
The intuitive reason why caching works is that data accesses in many domains of computer science exhibit a considerable degree of \emph{``locality''}.
A more formal way to capture this ``locality'' is to characterize the access frequency of all possible data items through a probability distribution and noting
that in many interesting domains of computer science,
the probability distribution is highly skewed, meaning that a small number of objects are
much more likely to be accessed than other objects.
Further, in many workloads, the access pattern, and consequently the corresponding probability distribution, change over time.
This phenomenon is also known as \emph{``time locality''}.

When a data item is accessed, if it already appears in the cache, we say that there is a \emph{cache hit}; otherwise, it is a \emph{cache miss}.
The ratio between the number of cache hits and the total number of data accesses is known as the \emph{cache hit-ratio}.
Hence, if the items kept in the cache correspond to the most frequently accessed items, then the cache is likely to yield a higher hit-ratio~\cite{LRU}.

Given that cache sizes are often limited, cache designers face the dilemma of how to choose the items that are stored in the cache.
In particular, when the memory reserved for the cache becomes full, deciding which items should be \emph{evicted} from the cache.
Obviously, eviction decisions should be done in an efficient manner, in order to avoid a situation in which the computation and space overheads required to
answer these questions surpasses the benefit of using the cache. The space used by the caching mechanism in order to decide which
items should be inserted into the cache and which items should be evicted is called the \emph{meta-data} of the cache. In many caching schemes,
the time complexity of manipulating the meta-data as well as the size of the meta-data are proportional to the number of items stored in the cache.

When the probability distribution of the data access pattern is constant over time, it is easy to show that the \emph{Least Frequently Used} (LFU)
yields the highest cache hit ratio~\cite{BCFPS99,ZipfCaching}.
According to LFU, in a cache of size $n$ items, at each moment the $n$ most frequently used items thus far are kept in the cache.
Yet, LFU has two significant limitations.
First, known implementations of LFU require maintaining large and complex meta-data.
Second, in most practical workloads, the access frequency changes radically over time.
For example, consider a video caching service; a video clip that is extremely popular on a given day might not be accessed at all a few days later.
Hence, there is no point in keeping that item in cache once its popularity has faded just because it was once very popular.

Consequently, various alternatives to LFU have been developed.
Many of these include aging mechanisms and/or focus on a limited size \emph{window} of last $W$ accesses.
Such aging is both in order to limit the size of the meta-data and in order to adapt the caching and eviction decisions to the more recent popularity of items.
A prominent example of such a scheme is called \emph{Window LFU (WLFU)}~\cite{WLFU}, as elaborated below.
Further, in the vast majority of cases, a newly accessed item is always inserted into the cache and the caching scheme focuses solely on its \emph{eviction policy}, i.e.,
deciding which item should be evicted.
This is because maintaining meta-data for objects not currently in the cache is deemed impractical.

Let us note that due to the prohibitively high cost of maintaining a frequency histogram for all objects ever encountered, published works that implement
the LFU scheme only maintain the frequency histogram w.r.t. items that are in the cache~\cite{SurveyOfCacheReplecmentStrategies}.
Hence, we distinguish between them by referring to the former as \emph{Perfect LFU} (PLFU) and the latter as \emph{In-Memory LFU}.
Since WLFU outperforms In-Memory LFU~\cite{BCFPS99} (at the cost of larger meta-data), here we only address WLFU and PLFU.

A popular alternative to LFU that relies on the ``time locality'' property is known as \emph{Least Recently Used} (LRU)~\cite{LRU}, by which the last item
accessed is always inserted into the cache and the least recently accessed item is evicted (when the cache is full).
LRU can be implemented much more efficiently than LFU\footnote{Yet, LRU is still considered too slow for hardware caches and operating systems page caching, and therefore in these highly demanding situations, we find approximations of LRU rather than exact LRU implementations.} and automatically
adapts to temporal changes in the data access patterns and to bursts on the workloads.
Yet, under many workloads, LRU requires much larger caches than LFU in order to obtain the same hit-ratio.

\subsubsection*{Contributions}
The main contribution of this paper is in showing the feasibility and effectiveness of augmenting caches with an approximate LFU based \emph{admission policy}.
This consists of the following aspects:
First, we present a caching architecture in which an accessed item is only inserted into the cache if an admission policy decides that the cache hit ratio is likely to benefit from
replacing it with the \emph{cache victim} (as chosen by the cache's replacement policy).
Second, we develop a novel highly space efficient data structure, nicknamed TinyLFU, that can represent an adaptable approximate LFU statistics for very large access \emph{sample}\footnote{
We use the term ``sample'' here rather than a ``window'' to highlight the role it plays in TinyLFU as well as to distinguish it from other references to ``windows'' in this work.
}.
For example, the meta-data required to implement TinyLFU, even for a domain of millions of items, can fit in a single memory page and thus can remain pinned in physical memory, allowing for fast manipulation.
We show through both a formal analysis and simulations that despite being approximated, the frequency ordering obtained by TinyLFU mimics closely the corresponding ordering that would appear in a true
access frequency histogram for these items.

Third, we report on the integration of TinyLFU into the Caffeine high performance Java cache open source project~\cite{CaffeineProject}.
In particular, this part includes another optimization of TinyLFU called \emph{W-TinyLFU} to better handle bursty workloads.

Fourth, we perform an extensive performance study of TinyLFU enhanced eviction policies and compare them to multiple other cache management policies including state of the art ones like LIRS~\cite{LIRS} and ARC~\cite{ARC,ARCpatent}.
We exhibit that for skewed workloads, both synthetic and real world traces, the TinyLFU admission policy dramatically boosts the performance of several known eviction policies.
In fact, in static distributions, after applying the TinyLFU admission policy, the difference between the various eviction policies becomes marginal.
That is, even the most naive eviction policy yields almost the same performance as true LFU.
In the case of real (dynamic) distributions, a smart eviction policy still has an impact, but much less profound than without the admission policy.
Further, on the various traces we have tested, W-TinyLFU's hit-ratio tops or equals all other cache management policies we have experimented with, and it is the only policy that performed so well consistently.
In fact, for most traces, even an LRU eviction policy enhanced with a TinyLFU admission policy obtained the same (best) results.

The rest of this paper is organized as follows:
Related work is discussed in Section~\ref{sec:related}.
We present TinyLFU in Section~\ref{sec:tinylfu}.
Section~\ref{sec:caffeine} describes the adaptation of TinyLFU into Caffeine~\cite{CaffeineProject} and our novel W-TinyLFU scheme.
The performance measurements are reported in Section~\ref{sec:results} and we conclude with a discussion in Section~\ref{sec:discussion}.

\section{Related Work}
\label{sec:related}

\subsection{Cache Replacement}

As indicated in the Introduction, while PLFU is an optimal policy when the access distribution is static,
the cost of maintaining a complete frequency histogram for all data items ever accessed is prohibitively high,
and PLFU does not adapt to dynamic changes in the distribution~\cite{WLFU,LFUAGING,LFUDA}. Consequently, several alternatives have been suggested.

In-Memory LFU only maintains the access frequency of data items already in the cache, and always inserts the most recently accessed
item into the cache, evicting, if needed, the least frequently accessed item among those that are in the cache~\cite{SurveyOfCacheReplecmentStrategies}.
In-Memory LFU is usually maintained using a heap.
The time complexity of managing LFU heaps was thought to be $O(\log(N))$ until recently, when an $O(1)$ construction was shown in~\cite{LFUIMPl}.
Yet, even with this improvement, In-Memory LFU still suffers from slow reaction to changes in the frequency distribution,
and its performance lags considerably compared to PLFU in static distributions, since it does not maintain
any frequency statistic for items that are no longer in the cache~\cite{SurveyOfCacheReplecmentStrategies}.

\emph{Aging} was introduced in~\cite{LFUAGING} to improve the ability of In-Memory LFU to react to changes.
It is obtained by limiting the maximum frequency count of cached items as well as occasionally dividing the frequency count of cached items by a given factor. Determining when to divide the counters and by how much is tricky and requires fine tuning~\cite{LFUDA}.

As mentioned above, WLFU only maintains the access frequency for a window of the last $W$ requests~\cite{WLFU}.
In order to maintain this window, the mechanism needs to keep track of the order of the requests, which adds an overhead.
Yet, WLFU adapts much better to dynamic changes in the access distribution than PLFU does.

ARC~\cite{ARC,ARCpatent} combines recency and frequency by maintaining meta data in two LRU lists.
The first reflects items that were only accessed once, while the second corresponds to items that were accessed at least twice in the recent history.
ARC adjusts itself to the characteristics of the workload by balancing its cache content source from both lists according to their hit-ratio.
In order to do it, ARC monitors access data for items that were recently evicted from the cache.
If items evicted from the LRU list are hit again, ARC extends this list and behaves more like LRU.
If an item that was recently evicted from the second LFU list is hit, ARC extends the size of the LFU list and behaves more similar to LFU.
The technique of monitoring items that were recently evicted from the cache is called \emph{ghost entries} and has become very common in state of the art cache policies.

\emph{LIRS}~\cite{LIRS} (\emph{Low Inter-reference Recency Set}) is a page replacement algorithm that attempts to directly predict the next occurrence of an item using a metric named \emph{reuse distance}.
To do so, LIRS monitors previous arrival times of many past items both in cache and ones that were evicted from the cache.
Thus, LIRS also maintains a large number of ghost entries.

The \emph{Segmented LRU (SLRU)}~\cite{SLRU} policy captures recent popularity by distinguishing between temporally popular items that are accessed at least twice in a short window vs. items accessed only once during that period.
This distinction is done by managing the cache through two fixed size LRU segments, a \emph{probation} segment (A1) and a \emph{protected} segment (A2).
New items are always inserted into the probation segment.
When an item already in the probation segment is being accessed again, it is moved into the protected segment.
If the probation segment becomes full, an insertion to it results in evicting (and forgetting) one of its items according to the LRU policy.
Similarly, any insertion to a full protected segment results in evicting its LRU victim.
However, in this case, the victim is inserted back into the probation segment rather than being completely forgotten.
This way, temporally popular items that reach the protected segment will stay longer in the cache than less popular items.

2Q~\cite{2Q} is a page replacement policy for operating systems, which manages two FIFO queues, A1 and A2.
A1 is split into two sub-queues, \texttt{A1in} (\emph{resident}) and \texttt{A1out} (\emph{non-resident}).
When an item is first accessed, it is entered into \texttt{A1in} and stays there until it is evicted under the FIFO policy.
However, when an item is being evicted from either \texttt{A1in} or A2, it is inserted into \texttt{A1out}.
If an accessed item is neither in \texttt{A1in} nor in A2, but it appears in \texttt{A1out}, then it is inserted into A2.

LRU-K~\cite{LRUK} also combines ideas from LRU and LFU.
In this policy, the last K occurrences of each element are remembered.
Using this data, LRU-K statistically estimates the momentary frequency of items in order to keep the most frequent pages in memory.

Web caching schemes often take into account also the size of objects.
For example, SIZE~\cite{SIZE} offers to simply remove the largest item first, while same size items are ordered by LRU.
LRU-SP~\cite{LruSP} weighs both the size and the frequency of an item when picking a cache victim.

GDSF~\cite{GDSF} is a hybrid Web caching policy that combines in its decisions the recency of last accesses, the cost of bringing the object to the cache, the object size, and its frequency.
It is an improvement of the Greedy Dual algorithm~\cite{GD} that only factors size and cost. 
In~\cite{CORBA}, latency is also taken into account in workloads when the data is hierarchical, i.e., in order to fetch a child item one must first fetch its predecessor item.

A related approach to ours is introduced in~\cite{AdaptiveCacheReplacement}, which suggests to augment known caching algorithms using a \emph{Hot List}.
This list indicates what the most popular items are using some decay mechanism.
Items in the hot-list are given priority over other items in the eviction policy.
However, the decision whether to evict an item in the cache does not depend on the relative frequency of this item vs. the frequency of the currently accessed item.
Further, an explicit list of $n$ items must be maintained.
This method was shown in~\cite{AdaptiveCacheReplacement} to somewhat improve the hit-ratio of various caching suggestions at the cost of significant meta-data overhead.

In summary, TinyLFU relates to the above suggestions by being a mechanism that can be used to augment other caching policies by enhancing them with approximate statistics on a large history while providing both quick access time and low meta-data overhead.

\subsection{Approximate Counting Architectures}

Approximate counting techniques are widely employed in many networking applications. Approximate counting was originally developed
in order to maintain a per stream network statistics. These constructions can be appealing to caching as well since they offer very
fast updates and consume a compact size.

Sampling methods~\cite{Sample1,Sample2,Sample3} offer a small memory footprint but require explicit representation of keys.
Also, they usually encounter relatively large error bounds. We therefore chose not to use them, as the size of the keys in our context is a
significant part of the overall costs.

Other methods such as Counter Braids~\cite{CounterBraids} reduce the meta data size but require a long decode operation and are therefore not applicable to caching.
Another approach is to compress the counter representation itself~\cite{SAC,DISCO,CEDAR,ICE-Buckets}.
In TinyLFU we manage to represent the histogram using short counters, and thus these methods do not help us.

Multi hash sketches such as \emph{Spectral Bloom Filters} (SBF)~\cite{SpectralBloom} and the \emph{Count Min Sketch} (CM-Sketch)~\cite{CountMinSketch} are appealing in our context.
As these sketches implicitly associate keys and counters, there is no need to store keys in the frequency histogram. Yet, they are not optimal for our case.
In particular, SBF includes a complex implementation aimed at achieving variable sized counters. Such a complex implementation is not needed in our case,
since we only require small counters.
The CM-Sketch on the other hand, offers a simple implementation, but is relatively inaccurate and therefore takes more space than \emph{Counting Bloom Filter}~\cite{CountingBloom}.

The \emph{Minimal Increment} scheme, also known as \emph{conservative update}, has been proposed as a way to boost the accuracy of Counting Bloom Filters when all the increments are positive~\cite{SpectralBloom,CUSketch,CUAnal}.
This technique was successfully used in several fields~\cite{USAGE1,USAGE2,USAGE3,USAGE4,USAGE5,USAGE6,Networking2LowPass,NLP1,NLP2}.

\subsection{Interesting caching applications}
Data caching can be used as a cloud service~\cite{Cloud1,Cloud2}. MemcacheD~\cite{MemCacheD} allows in memory caching of database queries and the items associated with them and is widely deployed by
many real life services including Facebook and Wikipedia.

Caching is also employed at the network level of data centers inside the routers themselves
using a technique called \emph{in-network caching}~\cite{NetworkCaching1,NetworkCaching2}. This technique allows content to be named explicitly,
and routers in the data center to cache data in order to increase the network capacity. In Peer-to-peer systems, TinyLFU was also used to expedite Kademlia lookup process~\cite{Shades}, effectively increasing capacity and reducing delays.

TinyLFU can be integrated in all the examples above. For in-network caching, TinyLFU is appealing since it requires very small memory overhead and can efficiently be implemented in hardware.
As for MemcacheD and cloud cache services, the eviction policy of~\cite{MemCacheD,Cloud2} is a variant of LRU with no admission policy.
As we show in this paper, adding a TinyLFU based admission policy to an LRU eviction policy greatly boosts its performance.

%


\section{TinyLFU Architecture}
\label{sec:tinylfu}

\subsection{TinyLFU Overview}
TinyLFU's architecture is illustrated in Figure~\ref{fig:architecture}.
Here, the cache eviction policy picks a cache victim, while TinyLFU decides if replacing the cache victim with the new item is expected to increase the hit-ratio.

To do so, TinyLFU maintains statistics of items frequency over a sizable recent history. Storing these statistics is considered prohibitively expensive for practical implementation and therefore TinyLFU
approximates them in a highly efficient manner. In the following, we describe the techniques we employ for TinyLFU,
some of which are adaptations of known \emph{sketching} techniques for approximate counting whereas others are novel ideas created especially for the context of caching.

\begin{figure}[t]
\center{
\includegraphics[scale=0.45]{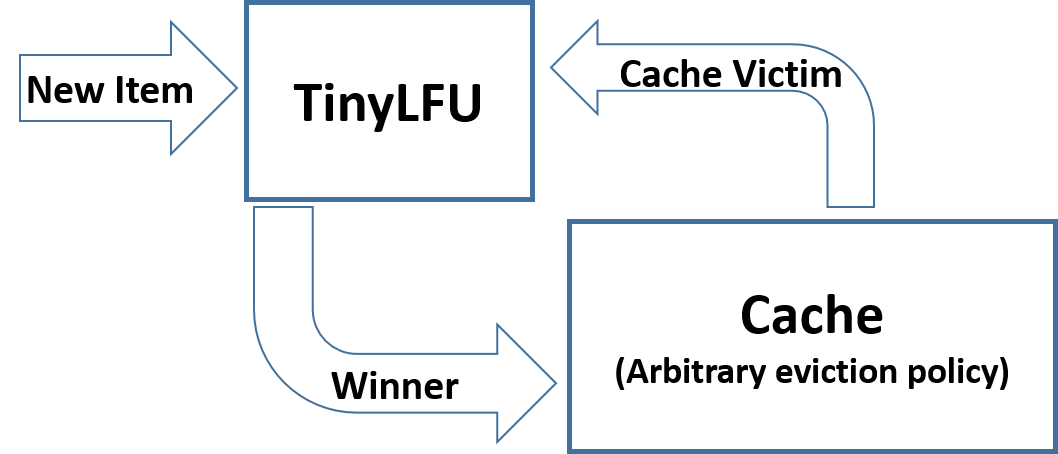}
}
\caption{A general cache augmented with TinyLFU}
\label{fig:architecture}
\end{figure}

Let us emphasize that 
we face two main challenges.
The first is to maintain a freshness mechanism in order to keep the history recent and remove old events.
The second is the memory consumption overhead that should significantly be improved in order to be considered a practical cache management technique.

\subsection{Approximate Counting Overview}

A \emph{counting Bloom filter} is a \emph{Bloom filter} in which each entry in the vector is a counter rather than a single bit.
Hence, rather than setting bits at the indexes corresponding to the filter's hash functions, these entries are incremented on an insert/add operation.

\begin{figure}[t]
\center{
\includegraphics[scale=0.4]{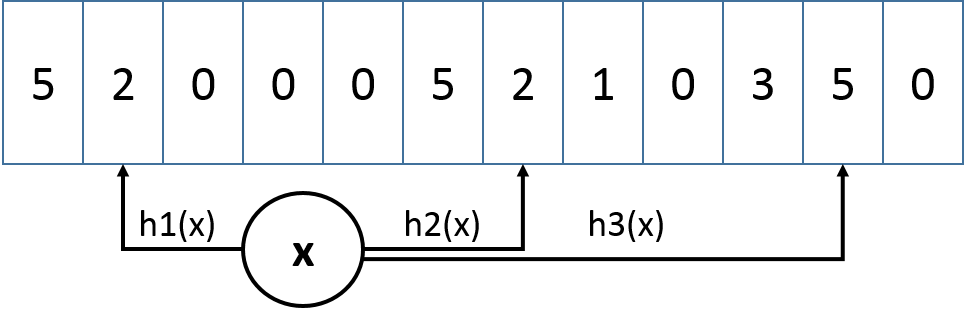}
}
\caption{A counting Bloom filter example}
\label{fig:MICBF}
\end{figure}

A \emph{Minimal Increment CBF} is an augmented counting Bloom filter that supports two methods: \texttt{Add} and \texttt{Estimate}.
The \texttt{Estimate} method is performed by calculating $k$ different hash values for the key.
Each hash value is treated as an index, and the counter at that index is read.
The minimal value of these counters is the returned value.
The \texttt{Add} method also calculates $k$ different hash values for the key.
However, it reads all $k$ counters and only increments the minimal counters.
For example, suppose we use 3 hash functions as illustrated in Figure~\ref{fig:MICBF}.
Upon item arrival, 3 counters are read.
Assuming we read $\{2,2,5\}$, the \texttt{Add} operation increments only the two left counters from 2 to 3 while the third counter remains untouched.
Intuitively, this \texttt{Add} operation prevents unnecessary increments to large counters and yields a better estimation for high frequency items, as their counters are less likely to be incremented by the majority of low frequency items.
It does not support decrements but is shown to empirically reduce the error for high frequency counts~\cite{SpectralBloom}.

\paragraph*{Other Approximate Counting Schemes:}
As mentioned above, CM-Sketch is another popular approximate counting scheme~\cite{CountMinSketch}.
It provides a somewhat lower accuracy per space tradeoff.
However, it is largely believed to be faster.
The space optimizations and aging mechanism we describe below can be applied to CM-Sketch just the same.
The description of TinyLFU below is oblivious to the choice between Counting Bloom Filter and CM-Sketch.



\subsection{Freshness Mechanism}
To the best of our knowledge, keeping approximation sketches fresh has only been studied in~\cite{SUNSHINE}, where a sliding sample is obtained by maintaining an ordered list of $m$ different sketches.
New items are inserted to the first sketch, and after a constant amount of insertions the last sketch is cleared and is moved to the head of the list.
This way, old events are forgotten. 

The method of~\cite{SUNSHINE} is not very appealing as a frequency histogram for two reasons:
First, in order to evaluate the frequency of an item, $m$ distinct approximation sketches are read, resulting in many memory accesses and hash calculations.
Second, this method increases the space consumption, since we need to store the same items in $m$ different sketches and each item is allocated more counters.

Instead, we propose a novel method for keeping the sketch fresh, the \emph{reset} method described below.
The reset operation is simple. Every time we add an item to the approximation sketch, we increment a counter.
Once this counter reaches the sample size ($W$), we divide it and all other counters in the approximation sketch by 2.
This division has two interesting merits. First, it does not require much extra space as its only added memory cost is a single counter of $Log(W)$ bits.
Second, this method increases the accuracy of high frequency items as we show both analytically and experimentally. Since the accuracy of an approximation
sketch can always be increased by using more space, we show that the reset method in fact reduces the total space cost since we get a significantly more
accurate sketch for the same space.

The downside of this operation is an expensive infrequent operation that goes over all the counters in the approximation scheme and divides them by 2.Yet, division by 2 can be implemented efficiently in hardware using shift registers.
Similarly, in software, shift and mask operations allow for performing this operation for multiple (small) counter at once. Finally, its amortized complexity is constant making it feasible for many applications.

In the following subsections, we prove the correctness of the reset operation and evaluate it's truncation error. The truncation error is caused since we use an integer division and therefor a counter that shows 3 will be reset to 1 and not to 1.5.
\subsubsection{Reset Correctness}
In this section, we analyze TinyLFU operations, under the assumption that the counting infrastructure works accurately. In the result section, we explore how much space is required to keep the performance identical to that of an accurate counting technique.

\begin{definition}
Denote $W$ the sample size, $f_i$ the frequency of item $i$, and $h_i$ the estimation of $i$ in TinyLFU.
\end{definition}
\begin{lemma}
Under constant distribution, at the end of each sample (immediately before each reset operation), the expected frequency of $i$ in the histogram is $E(h_i) =  f_i \cdot W$.
\end{lemma}
\begin{proof}
By induction on the number of reset operations performed $r$.

Base: For $r=0$, the condition holds trivially. We sample $W$ items one after another under constant distribution. By definition, item $i$ has a frequency of $f_i$. Therefore, the expected height of the histogram is $E(h_i)=f_i \cdot W$.

Step: Assume correctness for $r<j$ and prove it for $r=j$. From the induction hypothesis, until the $j-1$ reset operation occurs, $E(h_i) = f_i\cdot W$. Therefore,
immediately after the $j-1$ reset operation, $E(h_i) = \frac{f_i\cdot w}{2}$. There are exactly $\frac{W}{2}$ samples until the next reset operation and therefore $E(h_i)$
is expected to be incremented $\frac{f_i \cdot W}{2}$ times. Since the expectation is additive we conclude that right before the $j'th$ reset operation: $E(h_i) = f_i \cdot W$.
\end{proof}
The above lemma proves that TinyLFU converges to the right frequency under constant distribution.
To show that TinyLFU adjusts to changes, we prove a slightly stronger lemma, that says that under constant distribution eventually TinyLFU converges to the right frequency regardless of its initial value.

\begin{lemma}
	Under constant distribution, eventually $E(h_i)= f_i \cdot W$.
\end{lemma}
\begin{proof}
	Consider the initial value of $h_i$ right before a reset operation.
    We write this value as $f_i\cdot W + \sigma$ where $\sigma$ is the error.
    For example, if $h_i= 10$ but $f_i\cdot W = 8$, then $\sigma = 2$.
    After performing a reset operation, $h_i$ will be: $h_i = \frac{f_i\cdot W}{2} +\frac{\sigma}{2}$.
    The next reset operation is scheduled to happen within $\frac{W}{2}$ events and on average $h_i$ before the next reset operation will increase by: $\frac{f_i\cdot W}{2}$.
    It's new value will therefore be $h_i = \frac{f_i\cdot W}{2} +\frac{\sigma}{2}+  \frac{f_i\cdot W}{2}  = f_i\cdot W +\frac{\sigma}{2}$.
	We can repeat this process $k$ times and get that after $k$ more samples the value of $h_i$ is going to be:
	$h_i = \frac{f_i\cdot W}{2} +\frac{\sigma}{2^k}+  \frac{f_i\cdot W}{2}  = f_i\cdot W +\frac{\sigma}{2^k}$
	Therefore, if we take $k$ to infinity we get that:
	$${\lim _{k \to \infty }} (f_i\cdot W +\frac{\sigma}{2^k}) = f_i\cdot W$$
	In practice, since we use integer division, the initial error is eliminated after $log_2(\sigma)$ samples.
    However, in that case, there is also a truncation error as detailed below.
\end{proof}

\subsubsection{Reset Truncation Error}

Since our reset operation uses integer division, it introduces truncation error.
That is, after a reset operation, the value of a counter can be as much as $0.5$ lower than that of a floating point counter.
If we have to reset again, after the reset the truncation error of the previous reset operation is divided to $0.25$, but we accumulated a new truncation error of $0.5$ resulting in a total error of $0.75$.
It is easy to see that the worst case truncation error converges to at most one point lower than the accurate rate of the item.
Therefore, the truncation error affects the recorded occurrence rate of an item by as much as $\frac{2}{W}$ right after a reset operation.
This means that the larger the sample size, the smaller the truncation error.

\subsection{Space Reduction}
Our space reduction is obtained over two separate axes: First, we reduce the size of each of the counters in the approximation sketch.
Second, we reduce the total number of counters allocated by the sketch. These space saving optimizations are detailed below.

\subsubsection{Small Counters}
Naively implementing an approximation sketch requires using long counters.
If a sketch holds $W$ unique requests, it is required to allow counters to count until $W$ (since in principle an item could
be accessed $W$ times in such a sample), resulting in $Log(W)$ bits per counter, which can be prohibitively costly.
Luckily, the combination of the reset operation and the following observation significantly reduces the size of the counters.

Specifically, a frequency histogram only needs to know whether a potential cache replacement victim that is already in the cache is more popular than the item currently being accessed. However, the frequency histogram need not determine the exact ordering
between all items in the cache.
Moreover, for a cache of size $C$, all items whose access frequency is above $1/C$ belong in the cache (under the reasonable assumption
that the total number of items being accessed is larger than $C$). Hence, for a given sample size $W$, we can safely cap the counters by $W/C$.

Notice that this optimization is possible since our ``aging'' mechanism is based on the reset operation rather than a sliding window.
With a sliding window, in an access pattern in which some item $i$ alternates between $W/C$ consecutive accesses followed by $W/C+1$ accesses to other items, it could happen that $i$ would be evicted as soon as the sliding window shifts beyond the least recent $W/C$ accesses to $i$ even though $i$ is the most frequently accessed item in the cache.

To get a feel for counter sizes, when $W/C=8$, the counters require only 3 bits each, as the sample is 8 times larger than the cache itself.
For comparison, if we consider a small 2K items cache with a sample size of 16K items and we do not employ the small counters optimization, then the required counter size is 14 bits.

\subsubsection{Doorkeeper}

In many workloads, and especially in heavy tailed workloads, unpopular items account for a considerable portion of all accesses.
This phenomenon implies that if we count how many times each unique item in the sample appeared, the majority of the counters are assigned to items that are not likely to appear more than once inside the sample.
Hence, to further reduce the size of our counters, we have developed the \emph{Doorkeeper} mechanism that enables us to avoid allocating multiple-bits
counters to most tail objects.

The Doorkeeper is a regular Bloom filter placed in front of the approximate counting scheme.
Upon item arrival, we first check if the item is contained in the Doorkeeper.
If it is not contained in the Doorkeeper (as is expected with first timers and tail items), the item is inserted to the Doorkeeper and otherwise, it is inserted to the main structure.
When querying items, we use both the Doorkeeper and the main structures.
That is, if the item is included in the Doorkeeper, TinyLFU estimates the frequency of this item as its estimation in the main structure plus 1.
Otherwise, TinyLFU returns just the estimation from the main structure.

When performing a reset operation, we clear the Doorkeeper in addition to halving all counters in the main structure.
Doing so enables us to remove all the information about first timers.
Unfortunately, clearing the Doorkeeper also lowers the estimation of every item by 1, which also increases the truncation error by 1.

Memory wise, although the Doorkeeper requires additional space, it allows the main structure to be smaller since it limits the amount of unique items that are inserted to the main structure.
In particular, most tail items are only allocated 1 bit counters (in the Doorkeeper).
Hence, in many skewed workloads, this optimization significantly reduces the memory consumption of TinyLFU.
TinyLFU and the Doorkeeper are illustrated in Figure~\ref{fig:tinyLFU}.
We note that a similar technique was previously suggested in the context of network security~\cite{DoorKeeper}.

\begin{figure}[t]
\center{
\includegraphics[scale=0.7]{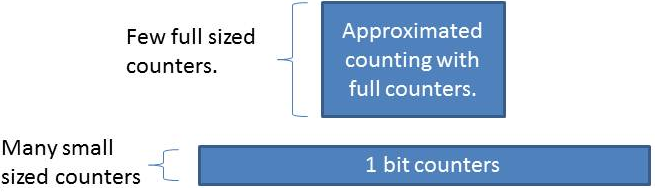}
}
\caption{TinyLFU structure}
\label{fig:tinyLFU}
\end{figure}

\subsection{Test Case: TinyLFU vs. a Strawman}
To motivate the design of TinyLFU, we compare it to a strawman.
The strawman is equivalent to building a frequency histogram using only existing approximate counting suggestions.
In order for the strawman to keep items fresh, it uses the sliding window approximation proposed in~\cite{SUNSHINE}.
That is, the strawman uses 10 different approximate counting sketches in order to mimic a sliding window.
Moreover, the strawman does not have a Doorkeeper or a cap on its counters and is therefore required to allow its counters to grow even to the maximal window size.

Our test case is a 1K items cache augmented by a 9K items frequency histogram.
For this workload, TinyLFU requires its counters to count up to 9.
This counting is obtained using 3 bits full counters that can count up to 8, in addition to the Doorkeeper that can count up to 1.
The strawman uses 10 approximate counting sketches of 900 items each, and its counters are of size 10 bits. 
In this example, we consider a Zipf 0.9 distribution, which is characteristics of many interesting real-world workloads, e.g.,~\cite{BCFPS99,ZipfCaching,ZipfYouTube}.

We summarize the storage requirements of both frequency histograms in Figure~\ref{fig:strawman}.
As can be observed, in this workload, TinyLFU is required to store approximately 10\% fewer unique values due to the fact that it uses a single big sketch instead of 10 small sketches.
Moreover, for this experiment the vast majority of items consumed only a small counter of 1 bit.
The frequent items that appeared more than once in this sample were allocated an additional 3 bits counters due to the small counters optimization.
In total, for this workload, we notice that TinyLFU reduces the memory consumption of the strawman by $\approx$89\%.

\begin{figure}[t]
\begin{center}
\scriptsize
\tabcolsep=0.11cm
\begin{tabular}{|c||c|c|c|c|c|}
\hline
 & \#Unique Items& \#2nd Timers&Full Size&Small Size&Average Size (bits)\tabularnewline
\hline
TinyLFU  & 7239  & 416 & 3  & 1 & 1.22\tabularnewline
\hline
Strawman & 8020  & -  & 10  & - & 10 \tabularnewline
\hline
\end{tabular}
\end{center}
\caption{ TinyLFU  vs unoptimized approximate frequency histogram}
\label{fig:strawman}
\end{figure}

\subsection{Connecting TinyLFU to Caches}
Connecting TinyLFU to both LRU and Random caches while treating the cache as a black box is trivial.
However, since TinyLFU uses a reset that halves the frequencies of items, we were forced to modify the implementation of the LFU cache in order to properly connect it with TinyLFU.
In particular, we had to synchronize the reset operation of TinyLFU with the cache and reset the frequencies of cached items as well.



\section{The W-TinyLFU Optimization}
\label{sec:caffeine}
TinyLFU was integrated into the Caffeine Java high performance caching library~\cite{CaffeineProject} that is available in open source.
During extensive benchmarking performed with this library, we have discovered that while TinyLFU performs well on traces originating from Internet services and artificial Zipf-like traces, there are a few workloads in which TinyLFU did not perform so well compared to state of the art caching policies.
This occurred mainly with traces that include ``sparse bursts'' to the same object, as is common in storage servers.
That is, in these cases, items belonging to new bursts do not manage to build enough frequency to remain in the cache before being evicted, causing repeated misses.

This problem was remedied in our Caffeine integration by devising a policy called \emph{Window Tiny LFU (W-TinyLFU)}, which consists of two cache areas.
The \emph{main cache} employs the SLRU eviction policy and TinyLFU admission policy while the \emph{window cache} employs an LRU eviction policy without any admission policy.
The A1 and A2 regions of the SLRU policy in the main cache are statically divided so that 80\% of the space is allocated to hot items (A2) and the victim is picked from the 20\% non hot items (A1).

Any arriving item is always admitted to the window cache and the victim of the window cache is given a chance to be admitted to the main cache.
If it is admitted, then the victim of W-TinyLFU is the main cache's victim and otherwise it is the window's cache victim.
The W-TinyLFU scheme is illustrated in Figure~\ref{fig:wtinylfu}.

\begin{figure}[t]
	\center{
		\includegraphics[scale=0.3]{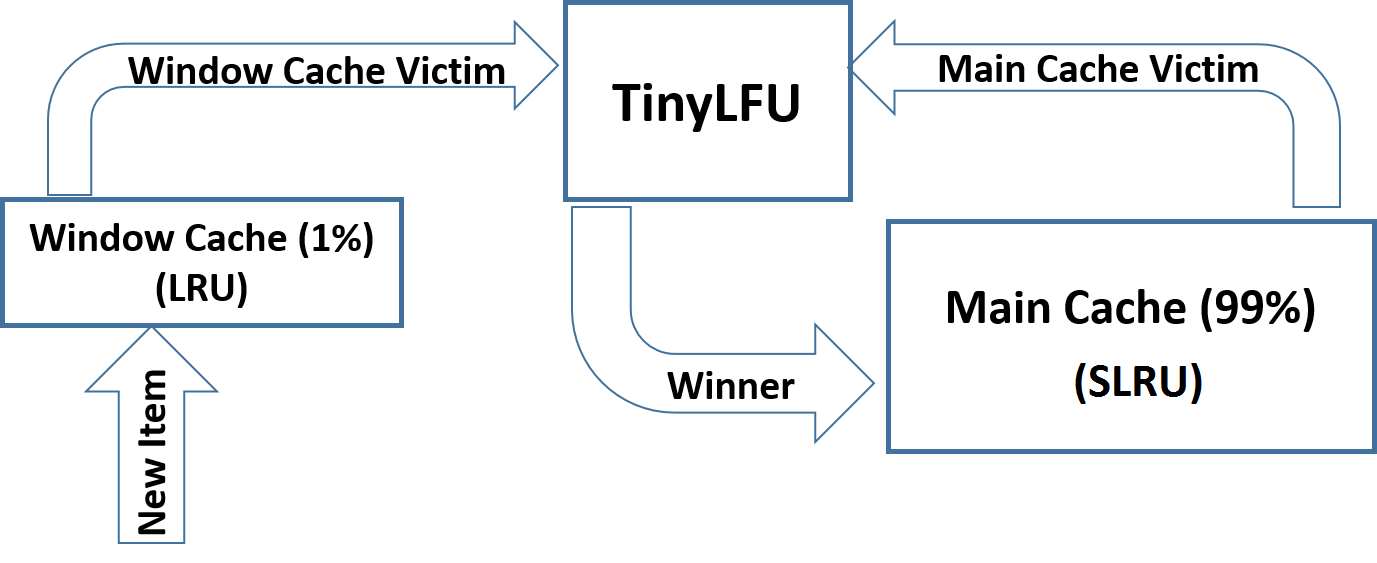}
	}
	\caption{Window TinyLFU scheme}
	\label{fig:wtinylfu}
\end{figure}

In the current release of Caffeine (2.0), the size of the window cache is $1\%$ of the total cache size and that of the main cache is $99\%$.
The motivation behind W-TinyLFU is to have the scheme behave like TinyLFU for LFU workloads while still be able to exploit LRU patterns such as bursts.
Because 99\% of the cache is allocated to the main cache (with TinyLFU), the performance impact on LFU workloads is negligible.
On the other hand, some workloads allow for exploitation of LRU friendly patterns.
In these workloads W-TinyLFU is better than TinyLFU.
As we report in Section~\ref{sec:results} below, for Caffeine's needs, W-TinyLFU is a top alternative for a wider variety of workload and thus the added complexity is justified.

TinyLFU was added into the Caffeine project as a late comer.
Our experience in doing so exposed several qualitative benefits for TinyLFU compared to other state-of-the-art schemes.
As we mentioned above, the latter typically maintain ghost entries, which adds complexity to the implementations.
Further, they involve a higher memory overhead.
The overall space overhead of W-TinyLFU in Caffeine is $8$ bytes per cache entry, which is significantly lower than the overhead ARC and LIRS.
Also, the development time was considerably faster in the case of TinyLFU (and W-TinyLFU) because the policy is less invasive, as there is no need to rework the main eviction scheme to account for ghost entries polluting the other data structure.

%




\section{Experimental Results}
\label{sec:results}

\subsection{Methodology}
\label{sec:method}
In this section, we examine the utility of TinyLFU as a practical cache admission policy.
Our evaluation is divided into three parts:
First, in Section~\ref{sec:simulator}, we examine the impact of augmenting three baseline eviction policies, namely LRU, Random and LFU, with a TinyLFU admission policy.
This part is performed using a Java prototype that can be instantiated with a given cache size (in terms of number of items it can store), a file that contains a trace of accesses, and one of the above cache management schemes.
The prototype then consumes the file, returning for each access whether according to the chosen management scheme and cache size it would have been a hit or a miss as well as overall statistics.
The prototype used Counting Bloom Filter with all optimizations described in Section~\ref{fig:architecture}.

Our second set of evaluations, reported in Section~\ref{sec:realruns}, is performed using the actual Caffeine implementation~\cite{CaffeineProject}.
In the current release of Caffeine (2.0), the TinyLFU histogram is maintained for a sample that is 10 times the cache size.
Also, Caffeine uses CM-Sketch (with our space optimizations and aging mechanism) instead of Counting Bloom Filters since even with CM-Sketch such a large sample only requires 8 bytes per cache item and on the other hand Caffeine's mission is high throughput\footnote{
The total meta-data size is larger since the LRU policy is maintained as a queue with forward and backward pointers adding two pointers for each in-memory item.}.
We have compared the obtained hit-ratios of a large number of cache management policies.
Yet, for brevity, we report only the best performing ones.
In the third set of measurements, reported in Section~\ref{sec:error}, we explore the various types of errors accumulated by TinyLFU's various mechanisms as well as the total error.

For the evaluations, we employ both constant distributions of Zipf 0.7 and Zipf 0.9\footnote{We have experimented with several other skewed distributions and obtained very similar qualitative results.}
and an SPC1-like synthetic trace~\cite{ARC} (denoted \textbf{SPC1-like} below) as well as real traces.
In the synthetic Zipf traces, items are picked according to the corresponding distribution from a set of 1 million objects.
Further, caches are given a long warm-up time (20 times the sample size) and we present them at their highest hit ratio.
The SPC1-like trace contains long sequential scans in addition to random accesses and has a page size of 4 Kbytes~\cite{ARC}.
In the real traces, the caches are not warmed up; since the distribution gradually changes over time, no warmup is necessary.

Our real workloads come from various sources and represent multiple types of applications, as detailed below:
\begin{itemize}
\item
\textbf{YouTube}: A YouTube trace is generated from a YouTube dataset~\cite{YouTubeDataSet}.
Specifically, the trace in~\cite{YouTubeDataSet} includes a weekly summary of the number of accesses to each video rather than a continuous time-line of requests.
Hence, for each reported week, we have calculated the corresponding approximate access distribution, and have generated synthetic accesses that follow this distribution on a week by week basis.

\item
\textbf{Wikipedia}: A Wikipedia trace containing $10\%$ of the traffic to Wikipedia during two months starting in September 2007~\cite{WikiTrace}.

\item
\textbf{DS1}: A data base trace taken from~\cite{ARC}.

\item
\textbf{S3}: A search engine trace taken from~\cite{ARC}.

\item
\textbf{OLTP}: A trace of a file system of an OLTP server~\cite{ARC}.
It is important to note that in a typical OLTP server, most operations are performed on objects already in memory and thus have no direct reflection on disk accesses.
Hence, the majority of disk accesses are the results of writes to a transaction log.
That is, the trace mostly includes ascending lists of sequential blocks accesses sprinkled with a few random accesses due to an occasional write replay or in-memory cache misses.

\item
\textbf{P8} and \textbf{P12}: Two windows server traces from~\cite{ARC}.

\item
\textbf{Glimpse}: A trace of the Glimpse system describing an execution of an analytic query~\cite{LIRS}.
The main characteristic of this trace is an underline loop, along with other types of accesses.

\item
\textbf{F1} and \textbf{F2}: Traces from applications running at two large financial institutions, taken from the Umass trace repository~\cite{UMAS}.
These are fairly similar in structure to the OLTP trace for the same reasons mentioned above.

\item
\textbf{WS1}, \textbf{WS2}, and \textbf{WS3}: Three additional search engine traces, also taken from the UMass repository~\cite{UMAS}.
\end{itemize}

Notice that the DS1, OLTP, P8, P12, and S3 traces were provided by the authors of ARC~\cite{ARC} and represent potential ARC applications.
The SPC1-like synthetic trace is also by the authors of ARC~\cite{ARC}.
Similarly, the Glimpse trace is provided by the authors of LIRS~\cite{LIRS}.




We have also used the YouTube workload to measure the impact of the dynamics of the distribution on the performance.
To do so, instead of playing the entire amount of views in a week, we only played a few samples from each week.
These tests simulate the behavior of the caches when the workload is very dynamic.

In our figures, LRU and Random refer to eviction policies while TLRU and TRandom refers to LRU and Random caches augmented by TinyLFU, respectively.
For LFU cache eviction, we tested two options: WLFU that uses both LFU eviction policy and LFU admission policy implemented using an accurate sliding window.
TLFU is the name we gave an LFU cache augmented with TinyLFU.
Further, ARC represents ARC and LIRS represents LIRS while W-TinyLFU stands for the W-TinyLFU scheme (1\% LRU window cache with and 99\% main cache managed with TinyLFU admission policy and an SLRU eviction policy).
In graphs where we explore different window cache sizes for W-TinyLFU, the window size is written in the parenthesis, e.g., W-TinyLFU(20\%) represents a window cache that consumes 20\% of the total cache size.


\subsection{Results of Augmenting Caches with TinyLFU}
\label{sec:simulator}

Figure~\ref{fig:PerformanceFilter} shows the results of TinyLFU admission policy on the performance of several eviction policies mentioned above under constant skewed distributions.
As can be seen, under constant distributions, all caches that are augmented with TinyLFU behave in a similar way.
Surprisingly, LFU cache eviction yields only a marginal benefit over the TinyLFU augmented LRU and Random techniques.
Let us note that in such skewed distributions, the maximal cache hit-ratio is theoretically bounded regardless of its size.
Intuitively, for a distribution function $f_i$, this bound can be roughly computed by the integral over $\max(0,f_i-1)$ (since the first occurrence is always a miss)
divided by the integral over $f_i$.

We conclude that for constant skewed distributions, the TinyLFU cache admission policy is an attractive enhancement.
While augmenting In-memory LFU yields slightly higher hit-ratio, the overheads of In-memory LFU may justify using a simpler cache eviction policy.
Particularly, LRU and even Random offer low overheads with comparable hit-ratios.


\begin{figure}[t]
\center{
\subfigure[Sample size is 32 times the cache size, Zipf 0.9]{\includegraphics[scale=0.68]{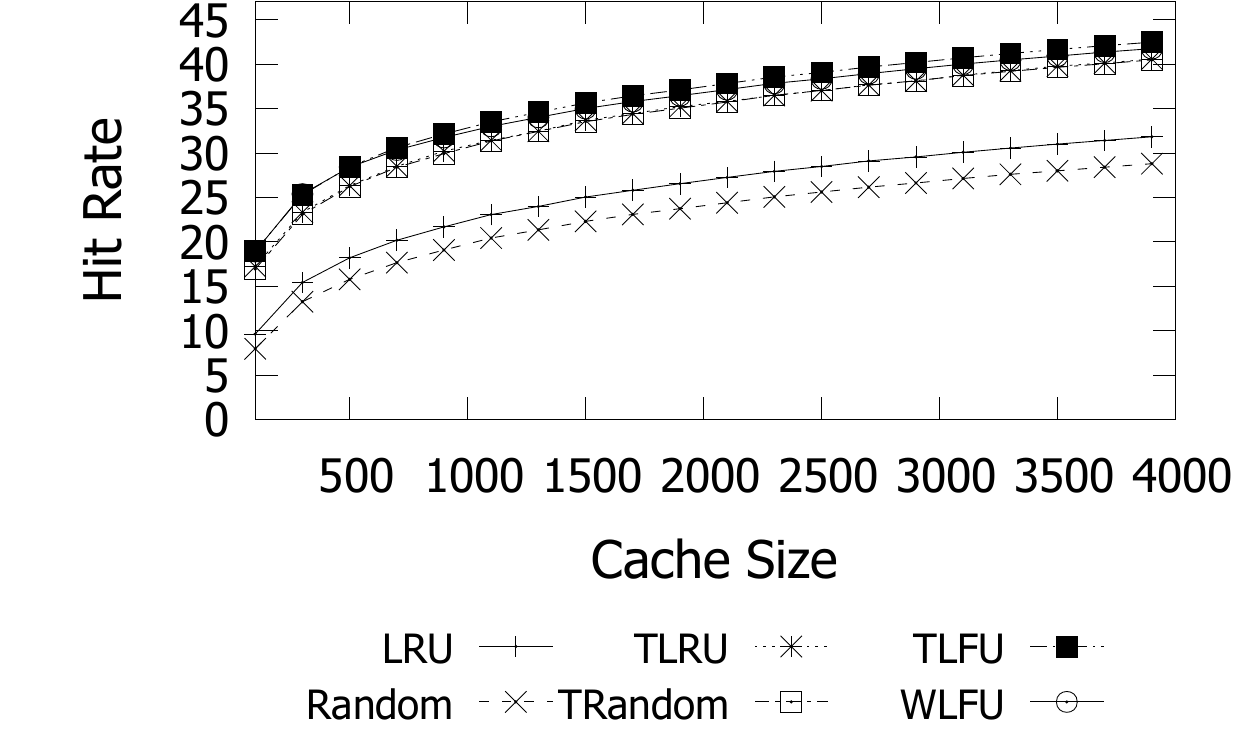}}
\subfigure[Sample size is 32 times the cache size, Zipf 0.7]{\includegraphics[scale=0.68]{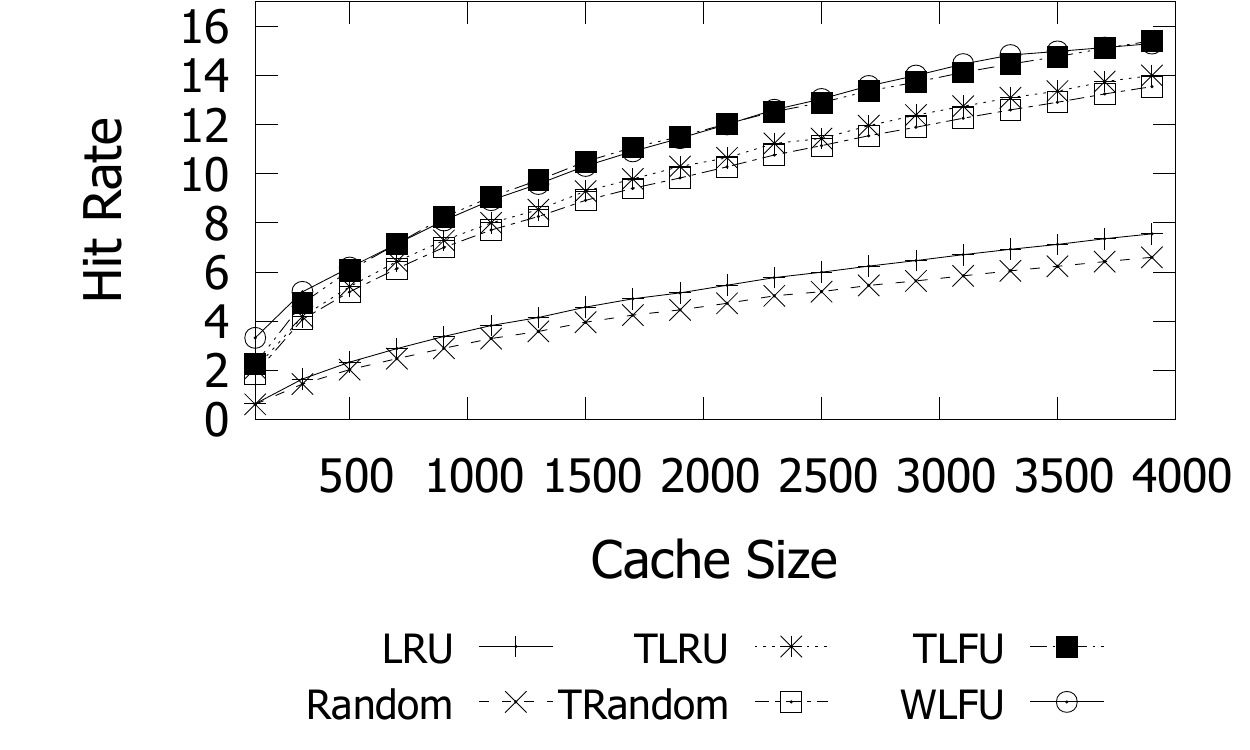}}
}
\caption{Augmenting arbitrary caches with a TinyLFU admission policy}
\label{fig:PerformanceFilter}
\end{figure}

\begin{figure}[t]
\center{

\subfigure[Effect of the change speed on the hit-ratio for 1000 items cache and sample size of 9000]{\includegraphics[scale=0.68]{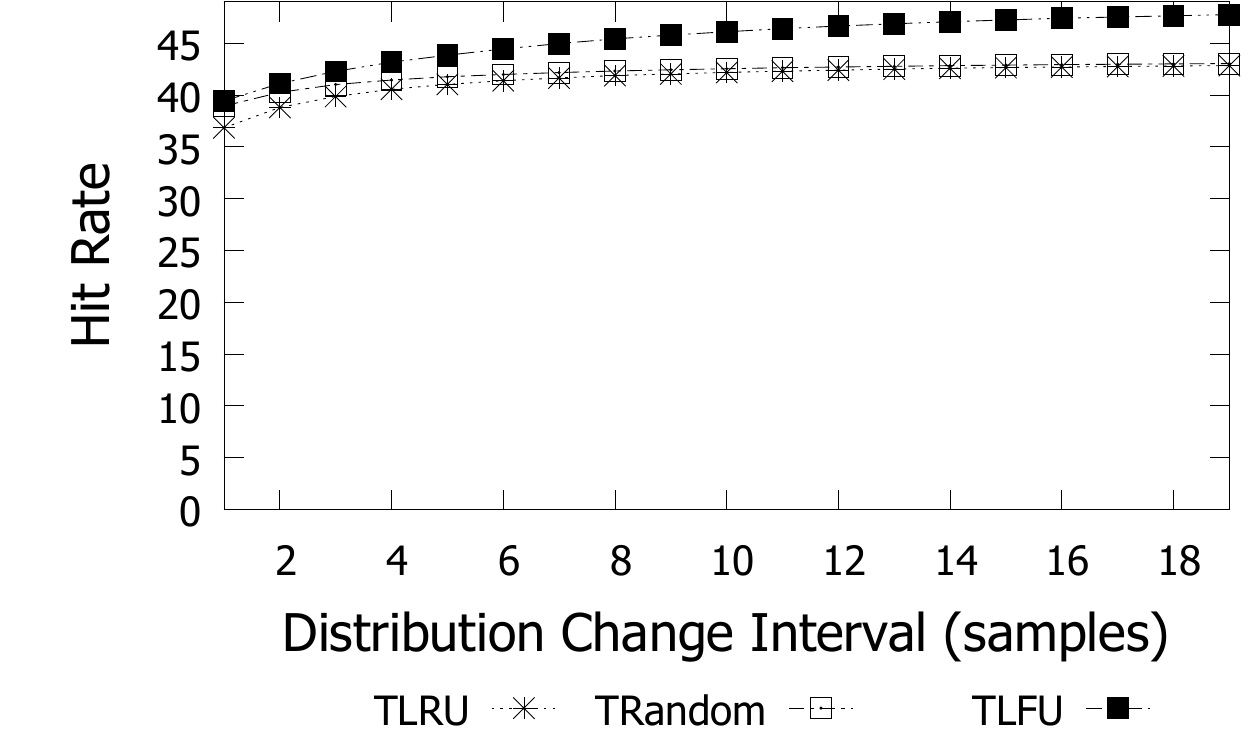}}
\subfigure[Cache size vs. hit-ratio for sample size of 9*CacheSize]{\includegraphics[scale=0.68]{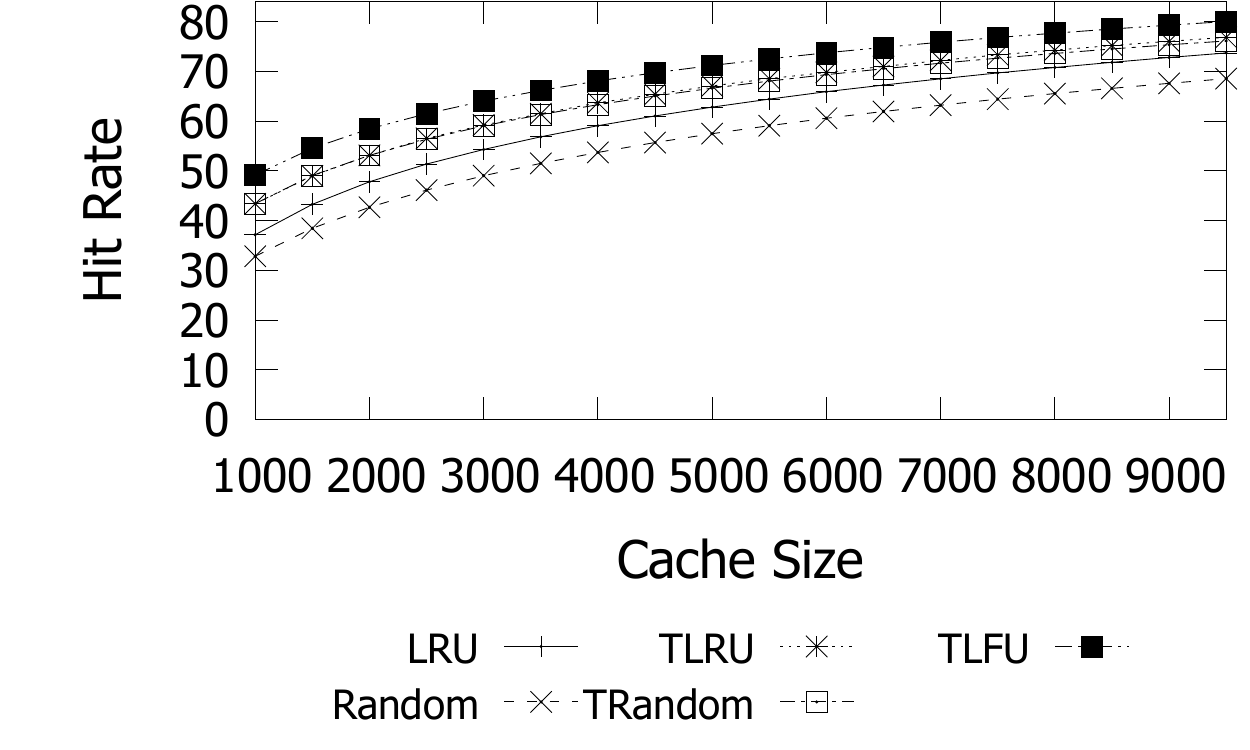}}
}
\caption{YouTube dataset}
\label{fig:YouTube}
\end{figure}

The second experiment we did was testing the augmented caches on a dynamic distribution. To do so, we used a dataset that describes the
 popularity of 161K newly created YouTube videos over 21 weeks starting at Apr. 16th, 2008 ~\cite{YouTubeDataSet}.
We evaluated the approximated frequencies of each of the videos each week and created a distribution that represents each week.
Our experiments therefore swap between these distributions every given amount of requests.

We measured two metrics: The first is how fast can we change the distributions, i.e., what is the effect of the number of samples we perform from each week's distribution
on the hit-ratio of our TinyLFU augmented caches.
The second is the impact of the cache size on the achieved hit-ratio when the distribution change speed is taken form the trace.

The results of this experiment are shown in Figure~\ref{fig:YouTube}. As can be observed, TinyLFU is effective in augmenting arbitrary caches even in dynamic workloads.
Further, the benefit is greater when the distribution changes more slowly, as expected from all LFU caches. Yet, in this workload, the difference
between an augmented Random cache and an augmented LRU cache to a true LFU cache is more significant.
Therefore, picking the correct cache victim seems to be more important in dynamic workloads than in static workloads.

The third measurement we made was running the experiment on the Wikipedia access trace.
We first studied the required ratio between sample size and cache size on samples of 100 million consecutive requests from different points in the trace.
Second, we used the best ratio we found and tested it on different cache sizes. These results are shown in Figure~\ref{fig:Wiki}. Unlike static workloads, real life workloads
gradually change over time. Therefore, using a very large sample can even reduce the obtained hit ratio as it slows the pace by which the cache adjusts to the workload.

We note that although WLFU and TLFU achieved nearly identical hit ratio, the main difference between WLFU and TLFU is in their meta-data costs.
For example, in the YouTube workload we used only $0.57$ bytes per sample element.
Since the TinyLFU statistics is maintained for a sample $9$ times the size of the cache, the total meta-data cost of TinyLFU is $9 \cdot 0.57 = 5.13$ bytes per cache entry.
If we consider that each cache entry should contain a video ID that requires 11 bytes, we conclude that TinyLFU is able to approximately remember a history 9 times bigger than the cache
with less space overhead than what is required to store just the keys of all cached items.

For comparison, WLFU is required to remember an explicit history 9 times bigger than the cache content.
Maintaining this history is expected to cost, even in the most space efficient implementation, 99 bytes per cache entry.
In addition, to operate quickly, it is required to maintain an explicit summary of these items, since iterating over the window and counting the frequency of cached items and replacement candidates is very slow.
Even if we neglect this additional memory overheads, WLFU's admission policy still requires almost 20 times more space than TinyLFU.

\begin{figure}[t]
\center{
\subfigure[\label{fig:windowSize}Sample size vs. hit-ratio for a 1000 items cache]{\includegraphics[scale=0.68]{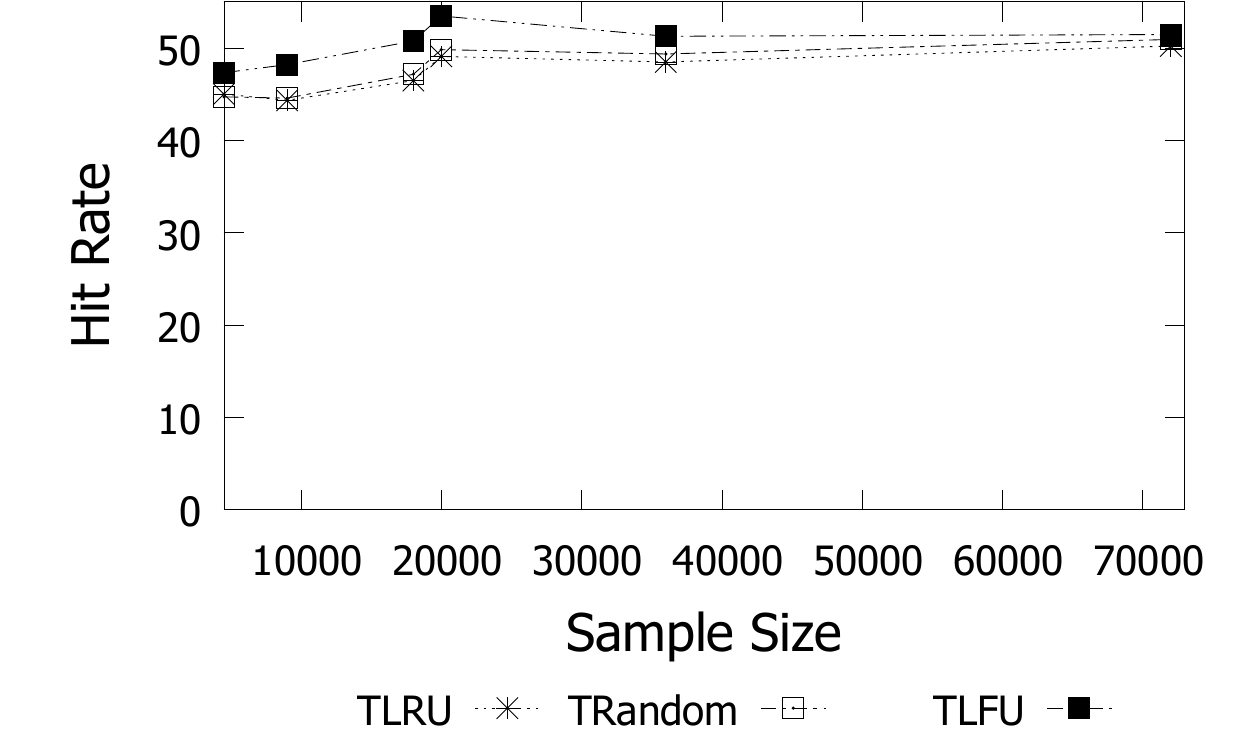}}
\subfigure[Cache size vs. hit-ratio for sample size of 20*CacheSize]{\includegraphics[scale=0.68]{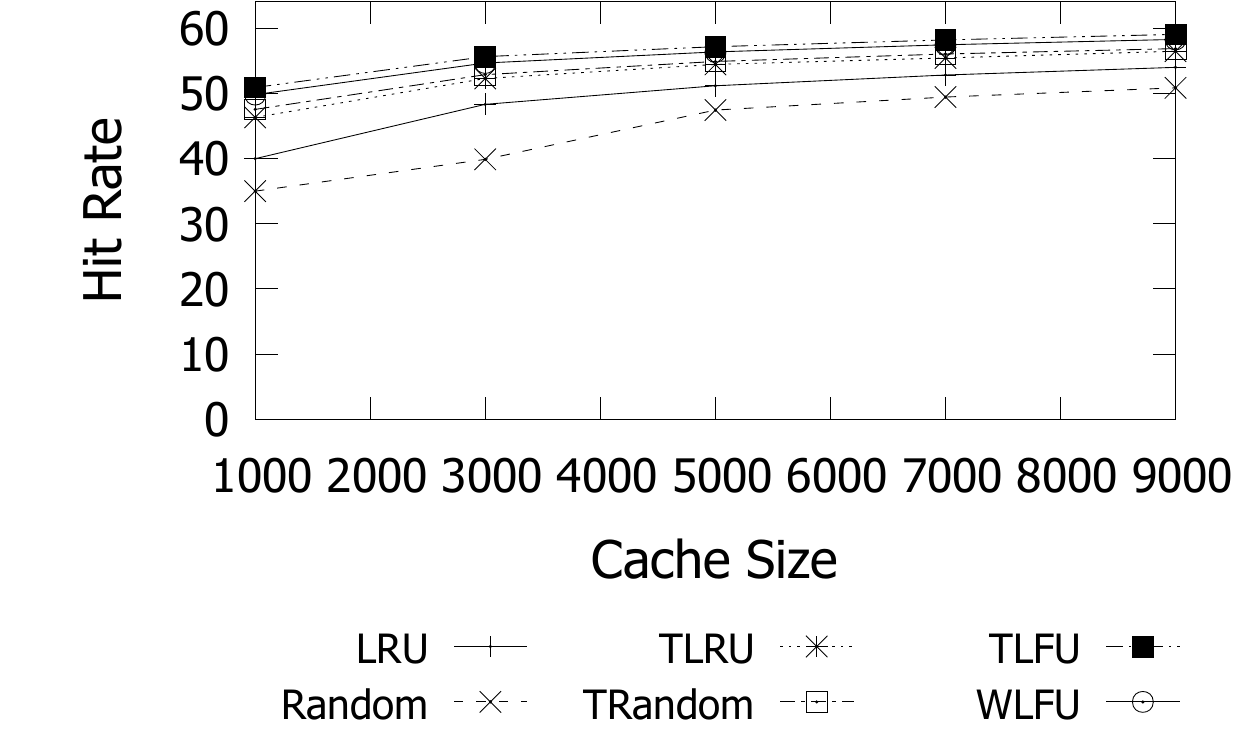}}
}
\caption{Evaluation over the Wikipedia trace}
\label{fig:Wiki}
\end{figure}

\subsection{Experiments with Caffeine}
\label{sec:realruns}

\subsubsection{Comparative Analysis}
\label{sec:comparing}

\paragraph*{Databases}
Figures~\ref{fig:glimpse} and~\ref{fig:ds1} present the hit-ratios of TLRU and W-TinyLFU compared to the state of the art schemes ARC and LIRS on the Glimpse~\cite{LIRS} and DS1~\cite{ARC} traces.
In both traces, TLRU and W-TinyLFU provide very good results.
In the Glimpse trace, they are at par with LIRS, while on the DS1 trace, both TLRU and W-TinyLFU outperform all other policies.
This indicates that the TinyLFU approach is a viable option for various types of databases.
\begin{figure}[t]
	\center{
		\includegraphics[scale=0.68]{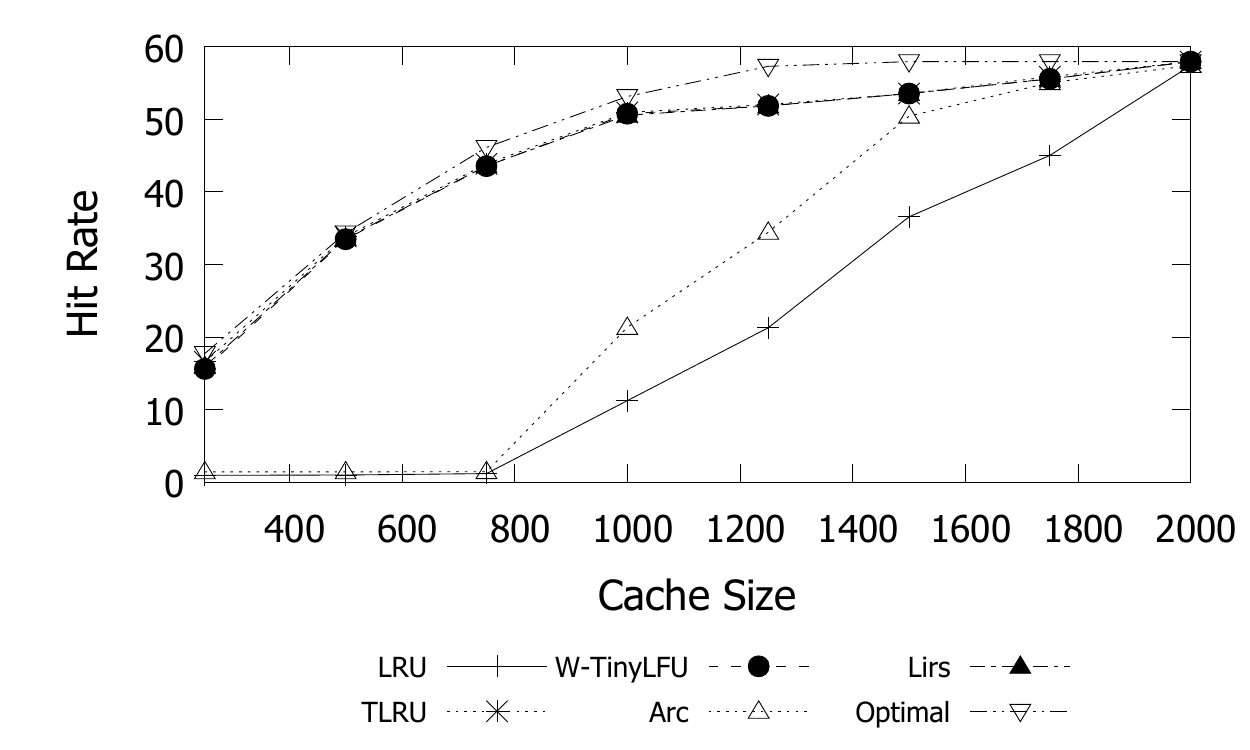}
    }
	\caption{Glimpse trace}
	\label{fig:glimpse}
\end{figure}
\begin{figure}[t]
	\center{
		\includegraphics[scale=0.68]{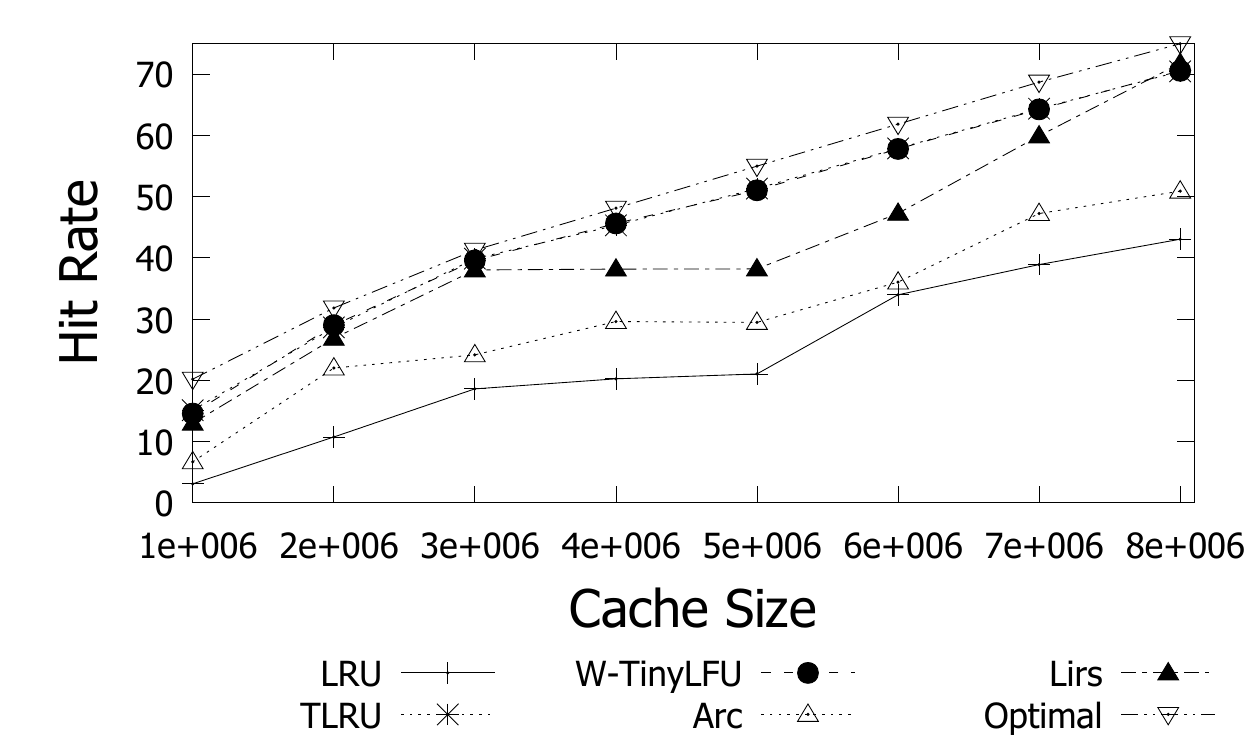}
	}
	\caption{DS1 database trace}
	\label{fig:ds1}
\end{figure}

\paragraph*{Windows File System}
Figures~\ref{fig:P8} and~\ref{fig:P12} shows results for the windows workstation traces P8 and P12~\cite{ARC}.
In these traces, we evaluate W-TinyLFU with both 1\% window cache and with a 20\% window cache.
As shown, both achieve very attractive hit ratios and outperform the other policies for all but the largest data points.
In Figure~\ref{fig:P12}, the graphs for the 20\% window cache and 1\% window cache intersect each other twice.
This indicates a potential benefit for being able to adaptively change the size of the window cache.

\begin{figure}[t]
	\center{
		\includegraphics[scale=0.68]{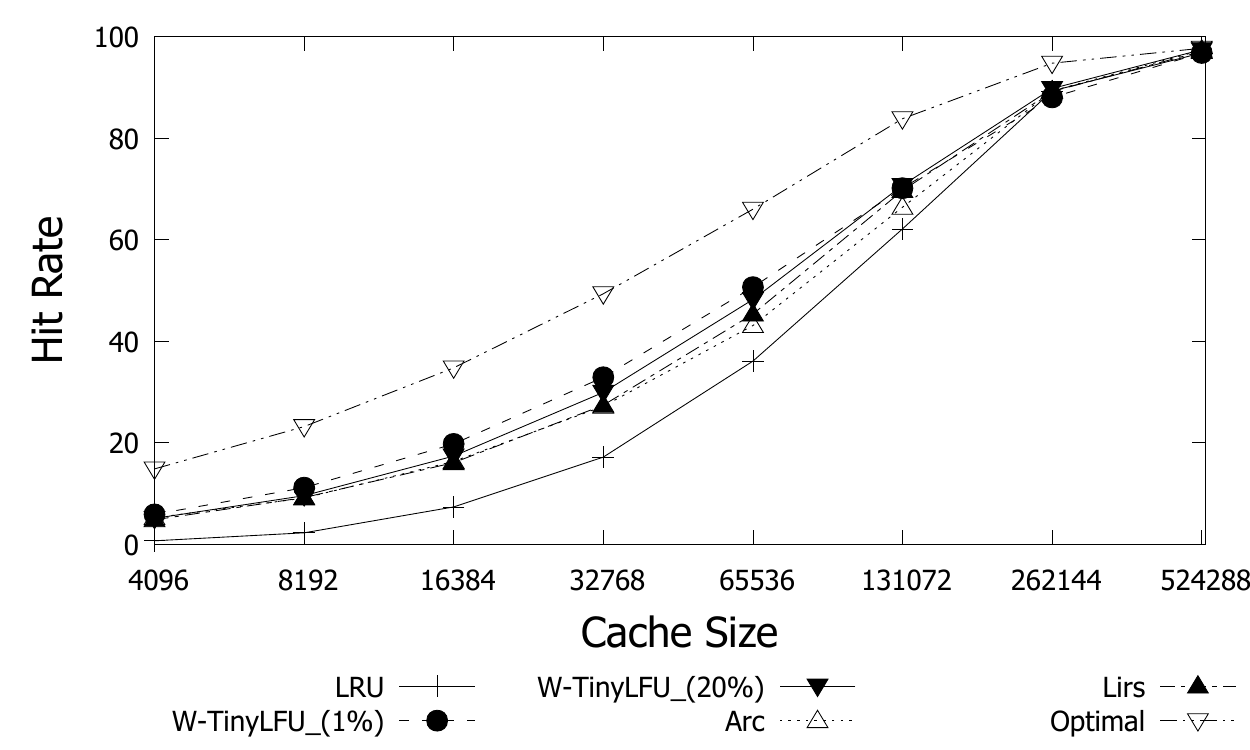}
	}
	\caption{P8 windows trace}
	\label{fig:P8}
\end{figure}
\begin{figure}[t]
	\center{
		\includegraphics[scale=0.68]{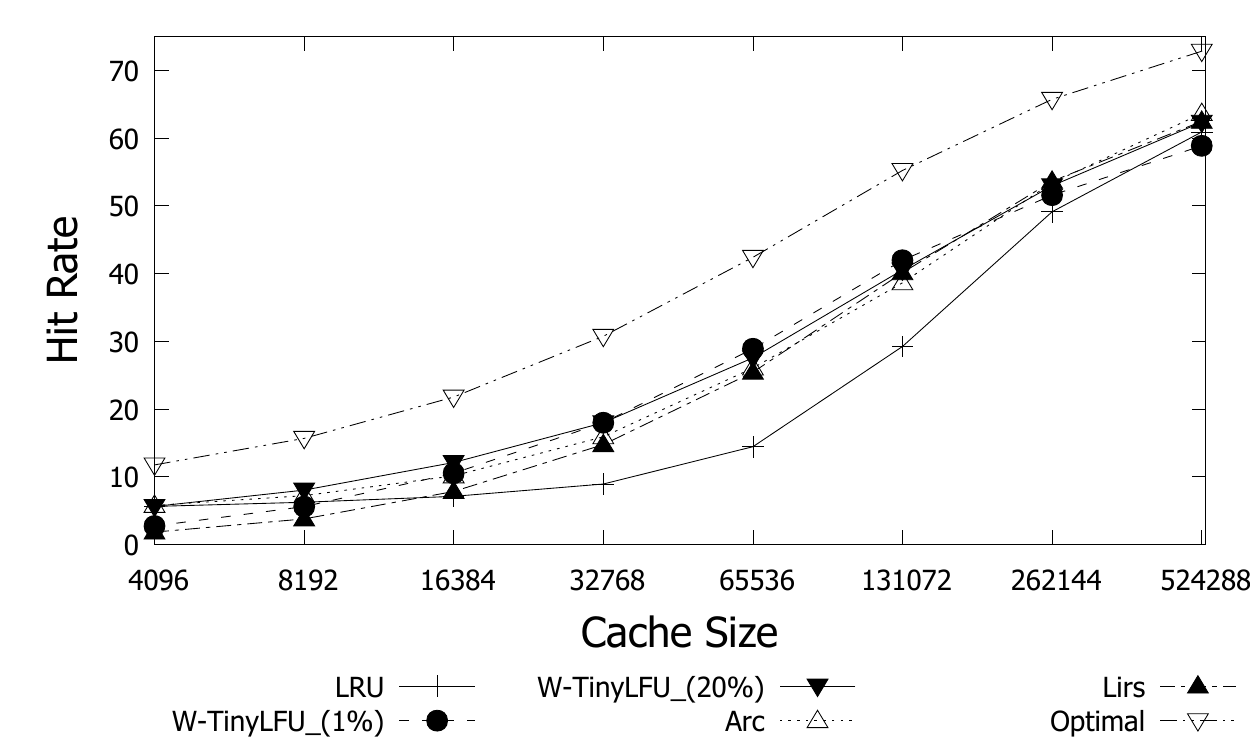}
	}
	\caption{P12 windows trace}
	\label{fig:P12}
\end{figure}

\paragraph*{Disk Accesses in OLTP Systems}
In the OLTP trace shown in Figure~\ref{fig:oltp}, TinyLFU behaves poorly due to the peculiarity of this trace as reported in Section~\ref{sec:method} above.
That is, ascending lists of sequential blocks accesses sprinkled with a few random accesses due to an occasional write replay or in-memory cache misses.
Therefore, TLRU is inferior to the rest.
Figure~\ref{fig:F1} and Figure~\ref{fig:F2} show results for two additional OLTP traces from servers running financial applications~\cite{UMAS}, exposing similar phenomena but with a lower intensity.
That is, in these two last figures the differences between all schemes are much smaller.

As can be seen, in all OLTP traces, W-TinyLFU's performance is comparable with the state of the art.
W-TinyLFU outperforms LRU and LIRS in all three traces.
Compared to ARC, W-TinyLFU is slightly better in the OLTP and F2 traces while in the F1 trace it is only better than ARC in smaller cache sizes with marginal differences between them.

We present here only W-TinyLFU(20\%) since it is slightly better than W-TinyLFU(1\%) on these traces and doing so reduces clobber in the graphs.
Section~\ref{sec:window-tuning} below explores the impact of the window cache size in W-TinyLFU on the hit ratio, where it is shown that a $20\%$ window cache is needed for best results with OLTP traces.
Unfortunately, increasing the window cache size reduces W-TinyLFU's performance with the other traces.
Hence, in the default configuration of the current version of Caffeine (2.0), the window cache size is fixed at $1\%$ of the total cache size in order not to burden the user with extra configuration parameters.

\begin{figure}[t]
	\center{
		\includegraphics[scale=0.68]{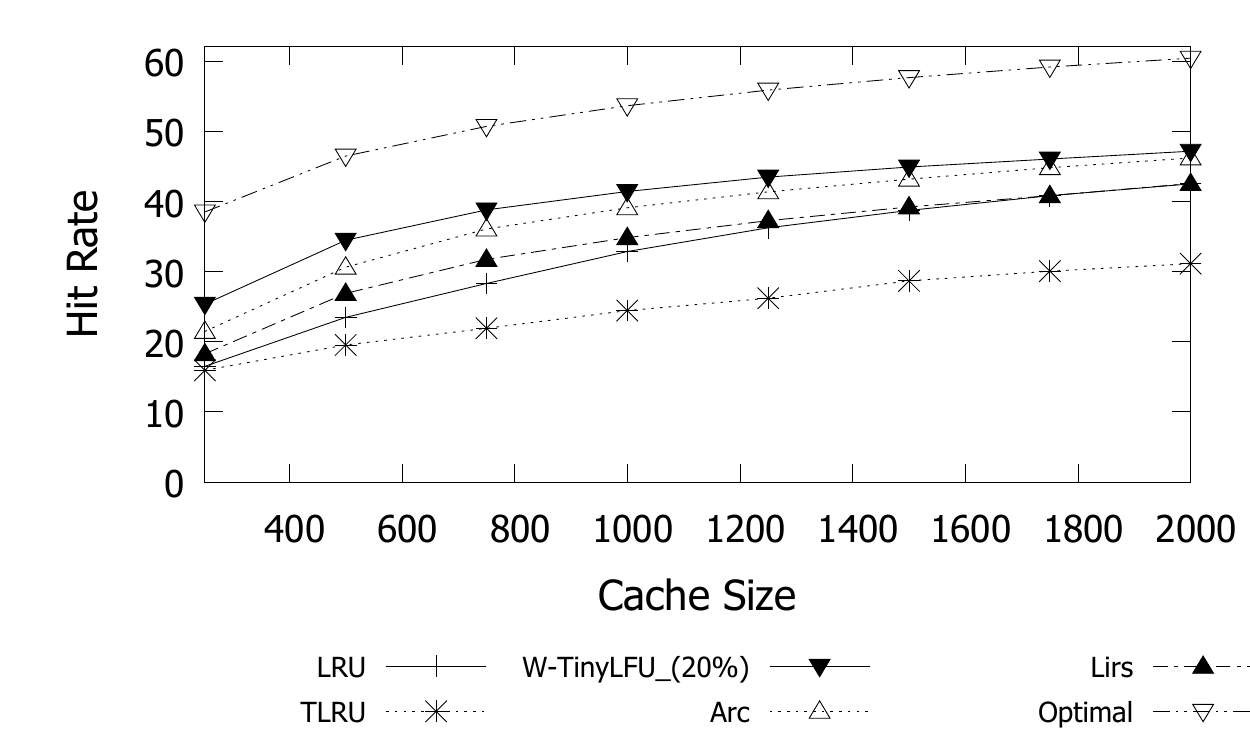}
	}
	\caption{OLTP trace}
	\label{fig:oltp}
\end{figure}
\begin{figure}[t]
	\center{
		\includegraphics[scale=0.68]{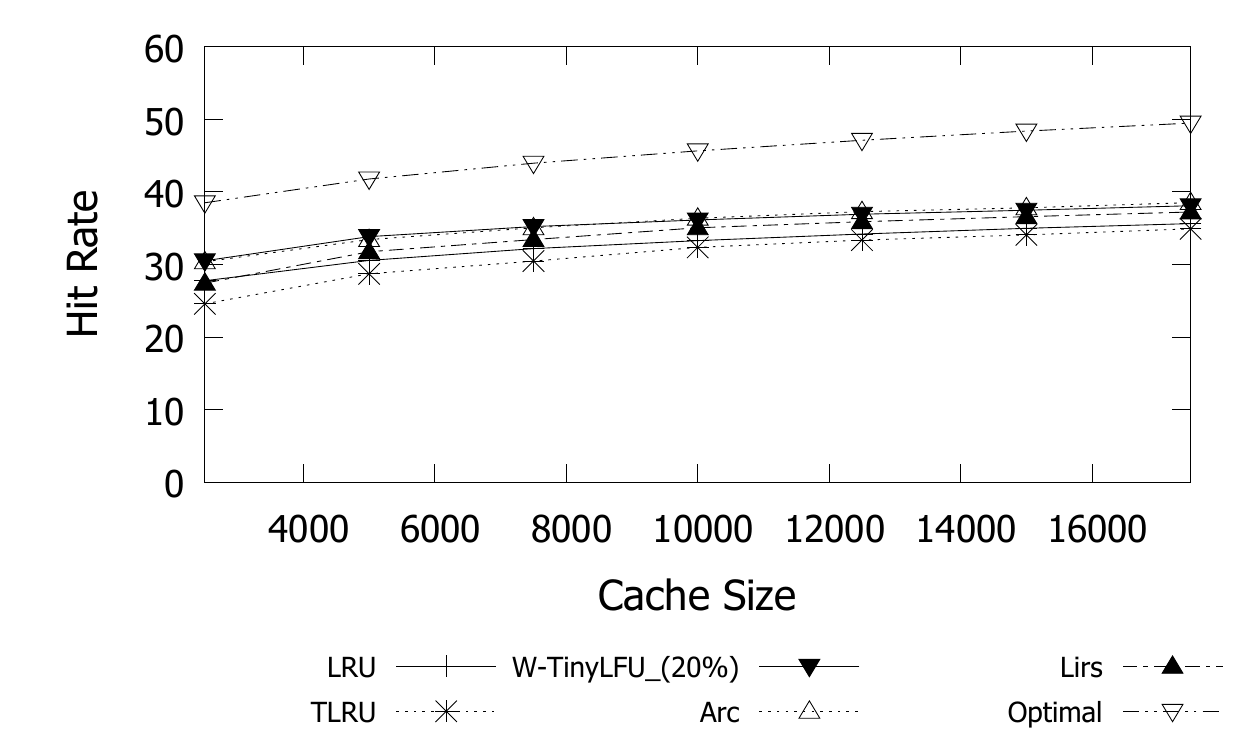}
	}
	\caption{Financial application (F1) trace}
	\label{fig:F1}
\end{figure}
\begin{figure}[t]
	\center{
		\includegraphics[scale=0.68]{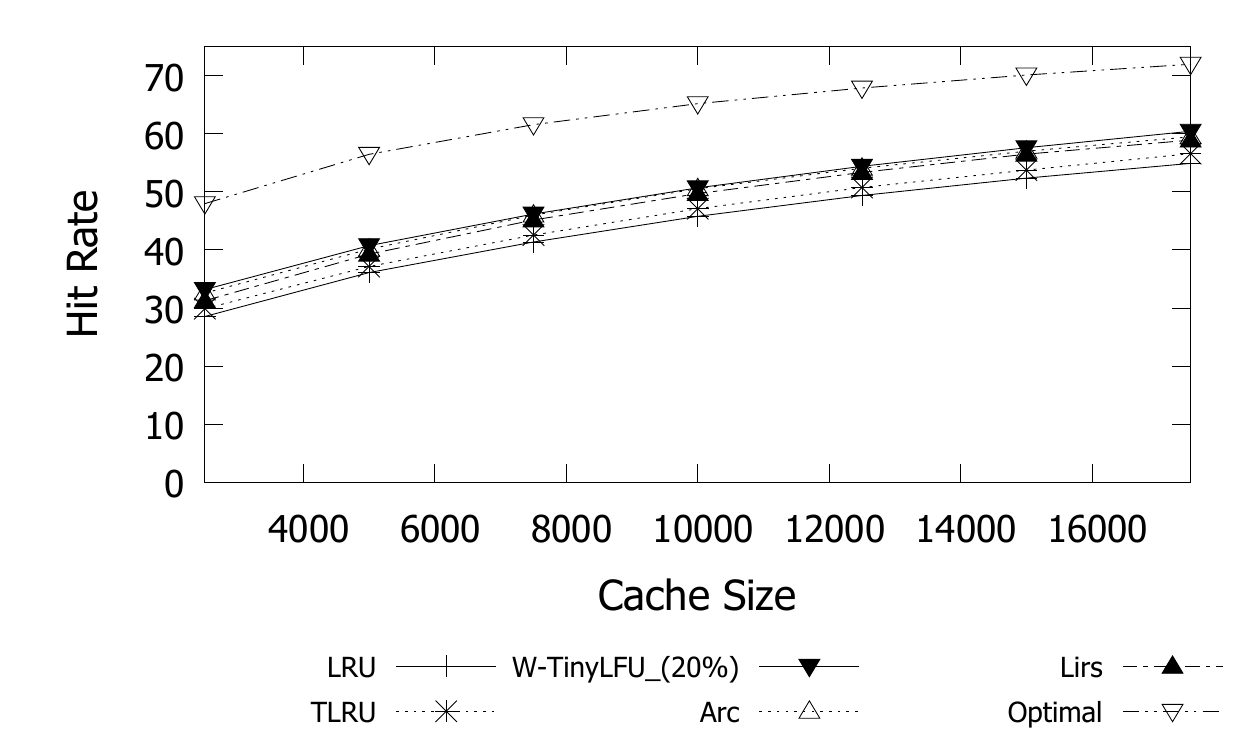}
	}
	\caption{Financial application (F2) trace}
	\label{fig:F2}
\end{figure}

\paragraph*{SPC1-like}
Figure~\ref{fig:SPClike} shows our result on the SPC1-like trace taken from~\cite{ARC}.
As can be observed, TLRU and W-TinyLFU outperform all other policies on that trace as well.
\begin{figure}[t]
	\center{
		\includegraphics[scale=0.68]{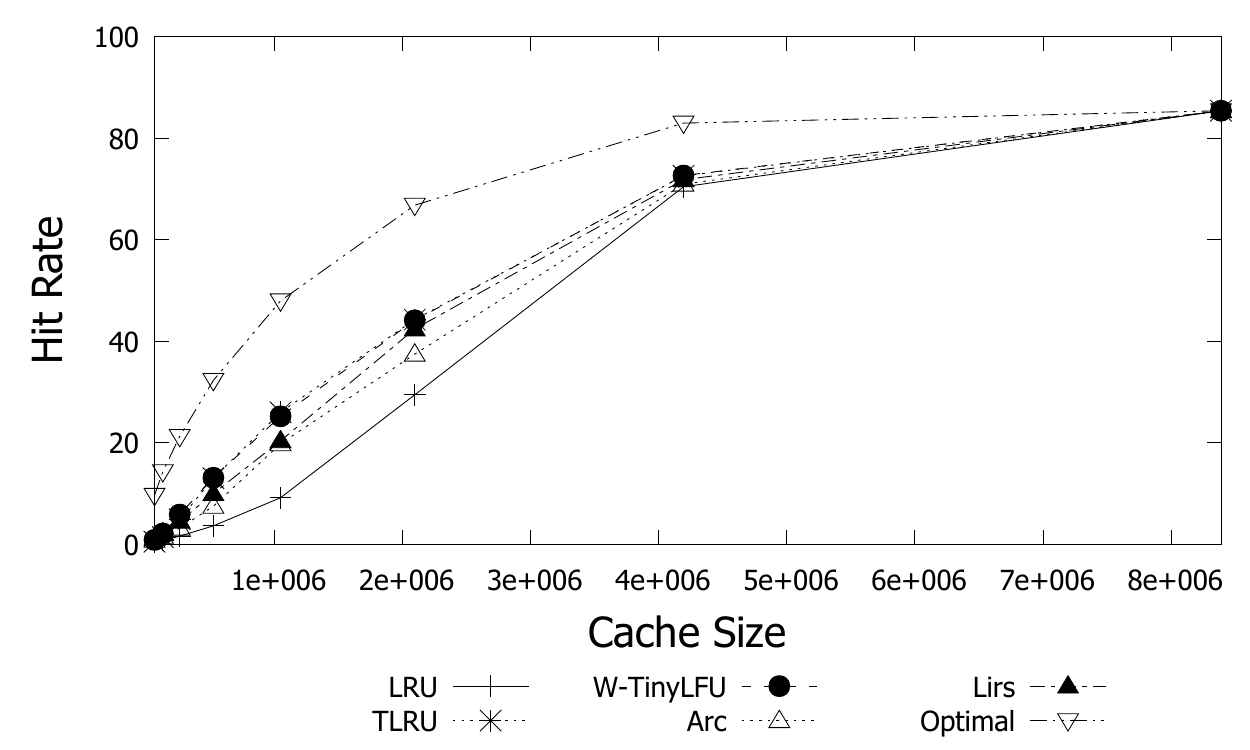}
	}
	\caption{SPC1-like trace}
	\label{fig:SPClike}
\end{figure}

\paragraph*{Search Engine Queries}
Figures~\ref{fig:search},~\ref{fig:search1},~\ref{fig:search2} and~\ref{fig:search3} demonstrate the behavior of W-TinyLFU and TLRU with a variety of search engine traces from~\cite{ARC,UMAS}.
As can be observed, between ARC and LIRS, ARC seems better for small caches and LIRS for large ones.
Yet, both TLRU and W-TinyLFU outperform them throughout.

\begin{figure}[t]
	\center{
		\includegraphics[scale=0.68]{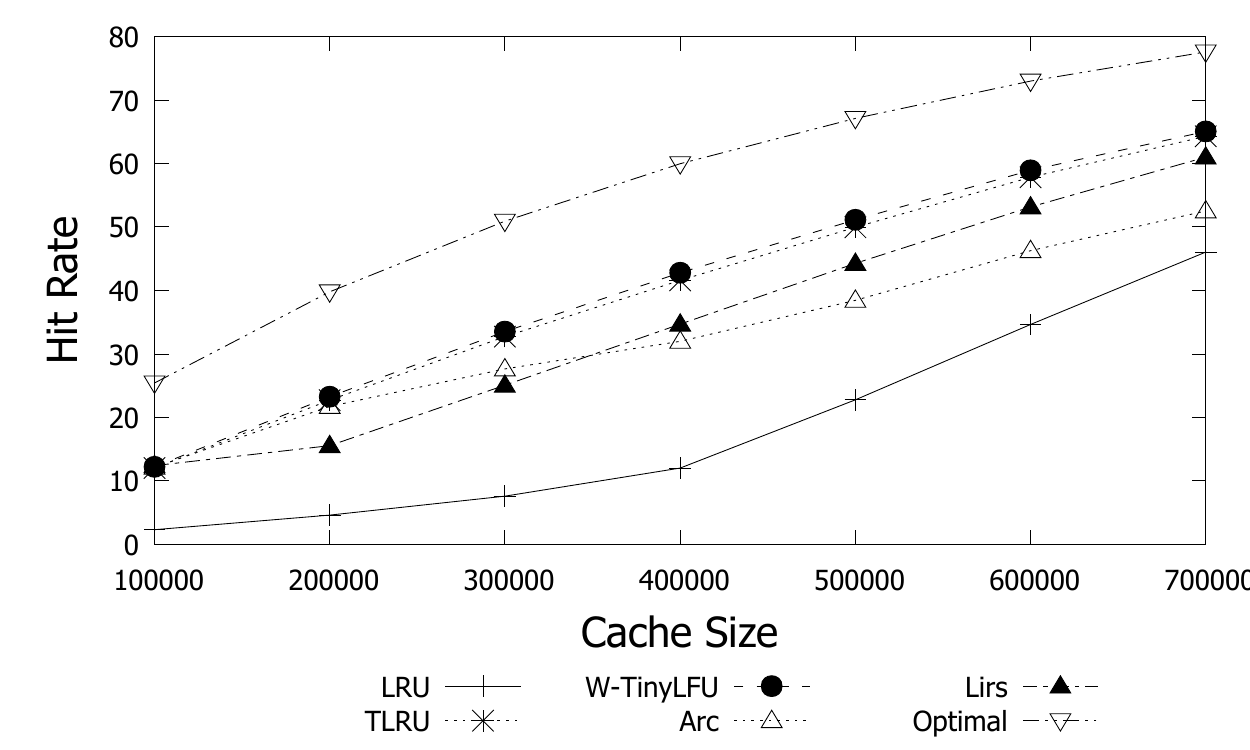}
	}
	\caption{Search engine trace(S3)}
	\label{fig:search}
\end{figure}
\begin{figure}[t]
	\center{
	\includegraphics[scale=0.68]{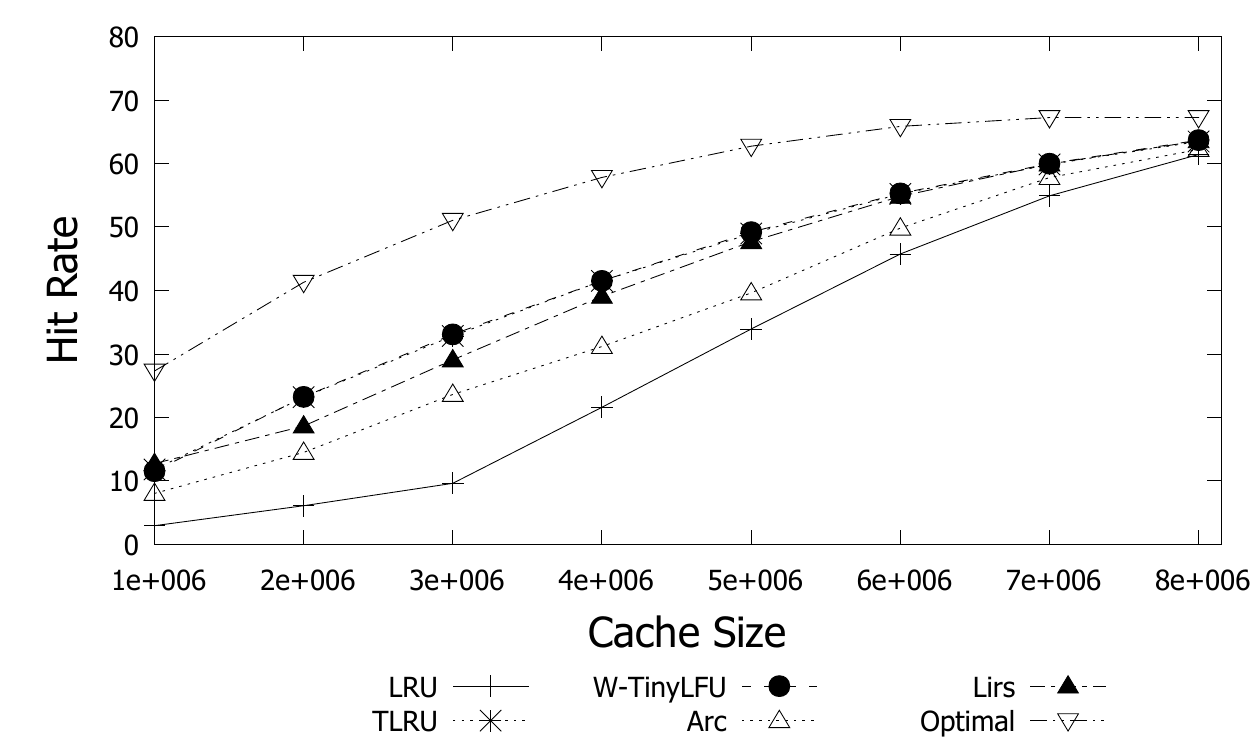}
	}
	\caption{Search engine trace(WS1)}
	\label{fig:search1}
\end{figure}
\begin{figure}[t]
	\center{
		\includegraphics[scale=0.68]{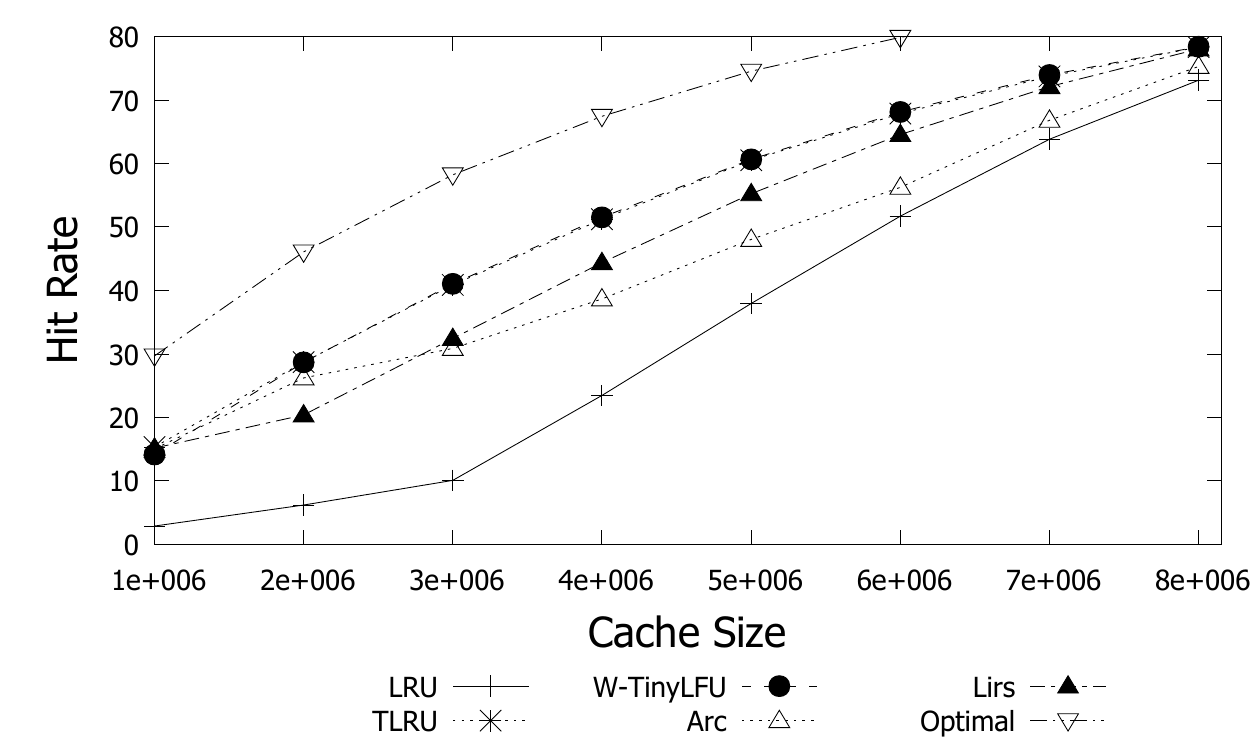}
	}
	\caption{Search engine trace(WS2)}
	\label{fig:search2}
\end{figure}
\begin{figure}[t]
	\center{
		\includegraphics[scale=0.68]{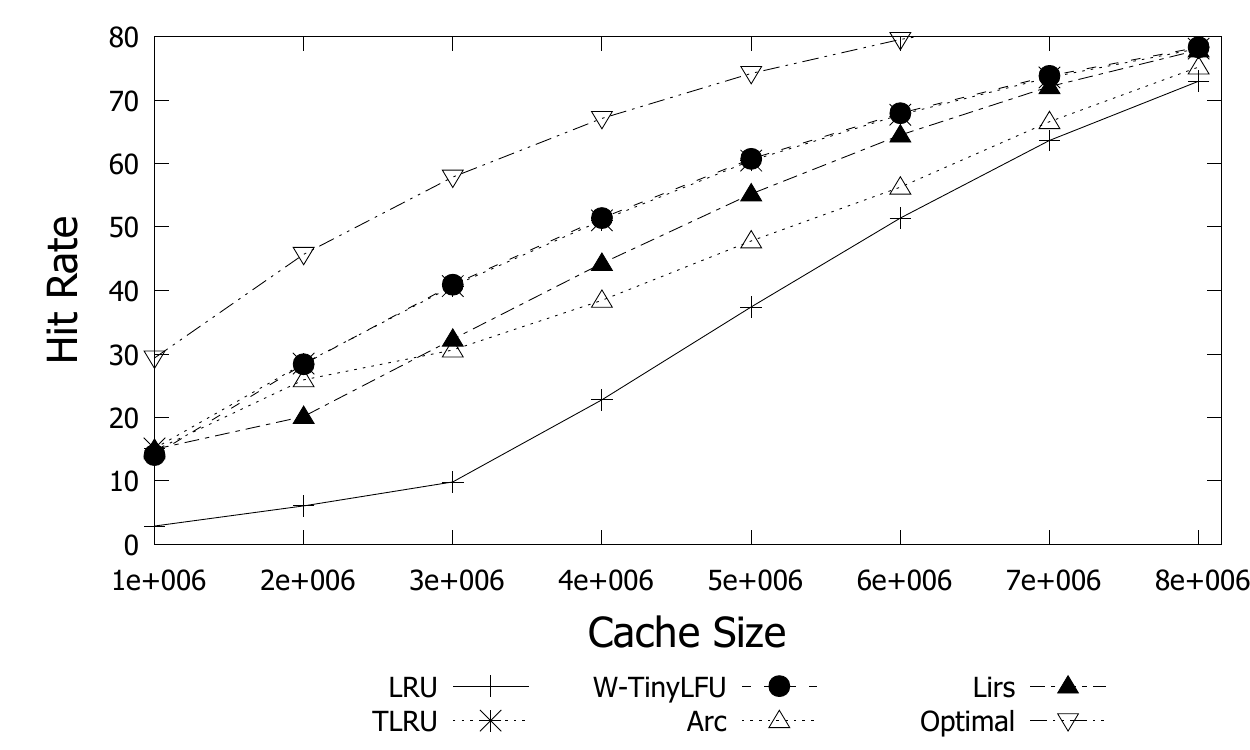}
	}
	\caption{Search engine trace(WS3)}
	\label{fig:search3}
\end{figure}

\subsubsection{Tuning the Window Cache Size in W-TinyLFU}
\label{sec:window-tuning}

As reported above, for the majority of traces, W-TinyLFU with a window cache whose size is $1\%$ of the total cache outperformed (or tied) all other schemes.
In contrast, in OLTP, F1, and F2, TinyLFU typically underperformed and while a window cache size of $1\%$ improved the hit ratio, it was still not the top performer.
In these traces, a larger window cache is needed.
Figure~\ref{fig:oltpWindowCache} explores the impact of the cache allocation between the window cache and the main cache on the hit ratio of W-TinyLFU under these three traces: OLTP, F1, and F2.
For OLTP we tested a total cache size of $1,000$ items while for F1 and F2 the total cache size is $2,500$ items.
As can be observed, window cache sizes of $20-40\%$ seem to perform best on these traces.

\begin{figure}[t]
 	\center{
 		\includegraphics[scale=0.68]{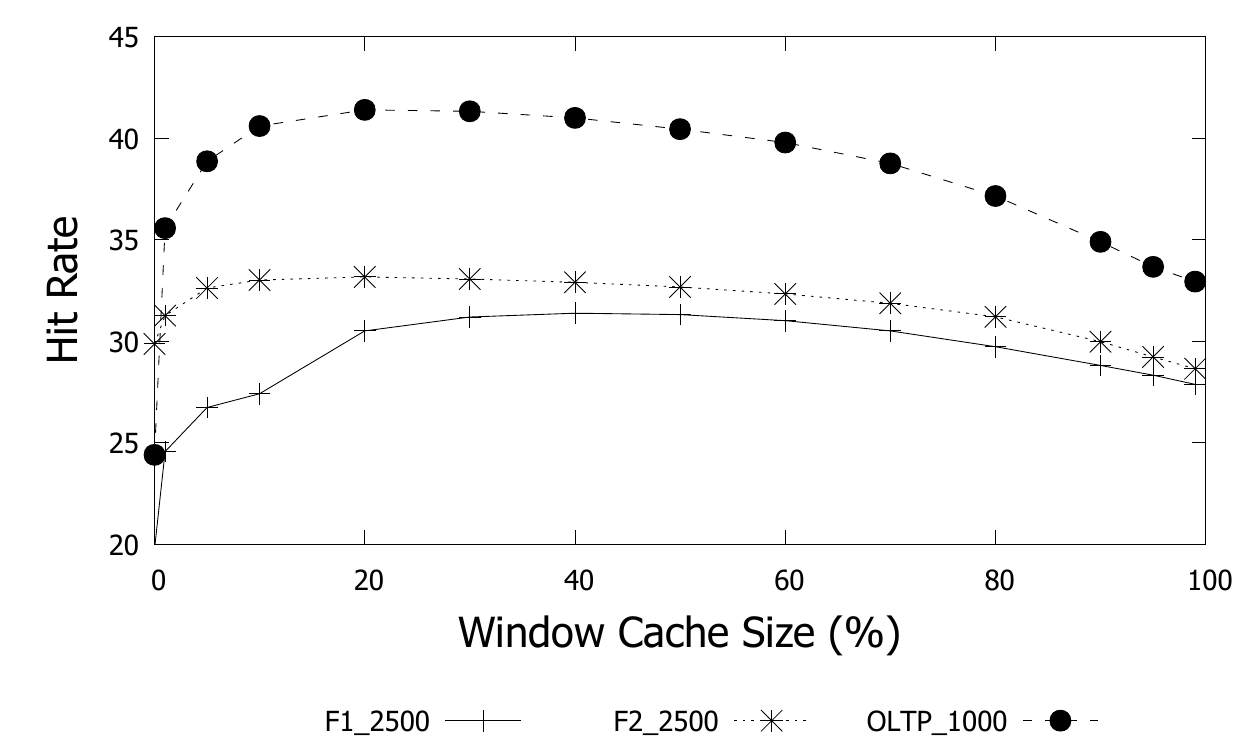}
 	}
 	\caption{Window/main cache balance for OLTP traces}
 	\label{fig:oltpWindowCache}
\end{figure}

\subsection{Sources of Inaccuracy}
\label{sec:error}

The most obvious source of inaccuracy in TinyLFU is the approximation accuracy introduced by false positives.
However, there is also the truncation error that is caused from performing integer division rather than a real (float) division in the reset operation.
For static distributions, it is also feasible to measure the sampling error, that is, how well TinyLFU was able to understand the underlining distributions.
All these factors together can explain the (in)accuracy of TinyLFU.

We define the \emph{Sampling Error} to be the difference between the hit ratio of an accurate TinyLFU with floating point counters to the ideal hit ratio achievable under a constant distribution.
The sampling error can be lowered for static distributions by using a larger sample size.

Next, the \emph{Truncation Error} is the change in hit ratio between an accurate TinyLFU (with no Doorkeeper and no false positives) and floating point counters to that of an accurate TinyLFU with integer counters.
The only difference between these two implementations is the type of division performed in the reset operation.
The truncation error can be lowered by allowing higher counts by counters (more bits per counter).

Finally, the \emph{Approximation Error} is the difference in hit ratio we obtain by using an approximated version of TinyLFU (as described in this paper) instead of an accurate one.
This means using a counting Bloom filter instead of a hash table.
The difference in hit ratio is due to false positives that may distort the sampling and make infrequent items appear frequent.

The results for a static Zipf 0.9 distribution are displayed in Figure~\ref{fig:summery}.
As can be observed, the approximation error does not appear at all until $\approx 1.25$ bytes per sample item.
That is, the approximate version of TinyLFU offers the same hit ratio as the accurate version.
As expected, the truncation error is smaller for small sample sizes and therefore in a $9$k items sample it is more dominant than in a $17$k items sample.
In the $17$k sample, the truncation error is almost negligible.

The sampling error is the most tricky one since increasing the sample size only helps with static distributions.
The sampling error gets smaller the larger we set the sample size.
This behavior is obviously not the case for real workloads, in which the sample size is to be chosen according to an empirical trial and error process.

Recall also that the TinyLFU version that was implemented in Caffeine employs CM-Sketch rather than a counting Bloom filter.
In its default configuration, it allocates $8$ bytes per cached entry and maintains a sample size that is $10$ times the cache size.
In other words, it uses $0.8$ bytes per element.
By doubling the allocated space to $16$ bytes per entry, the obtained hit ratio increases by approximately $\approx 0.5\%$ on the various traces.
In other words, the current configuration of Caffeine carries roughly a $0.5\%$ approximation error.


\begin{figure}[t]
\center{
\subfigure[9k items sample]{ \includegraphics[scale=0.7]{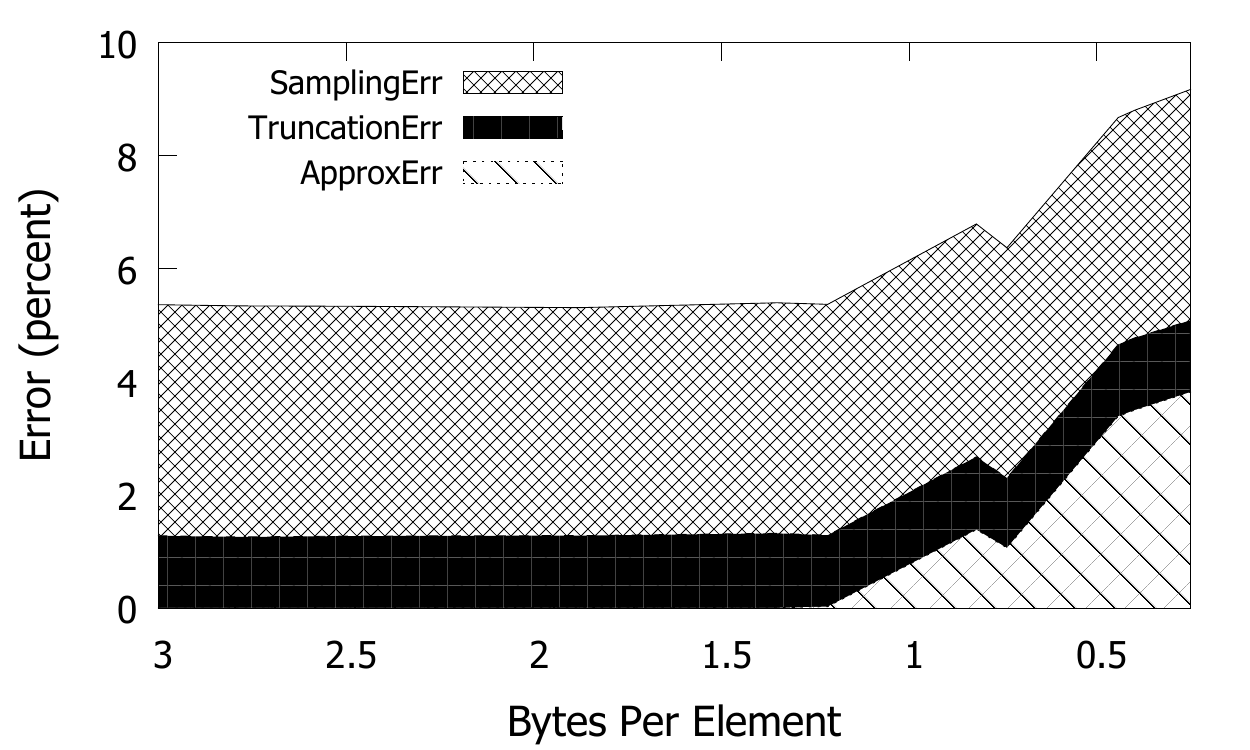}}
\subfigure[17k items sample]{\includegraphics[scale=0.7]{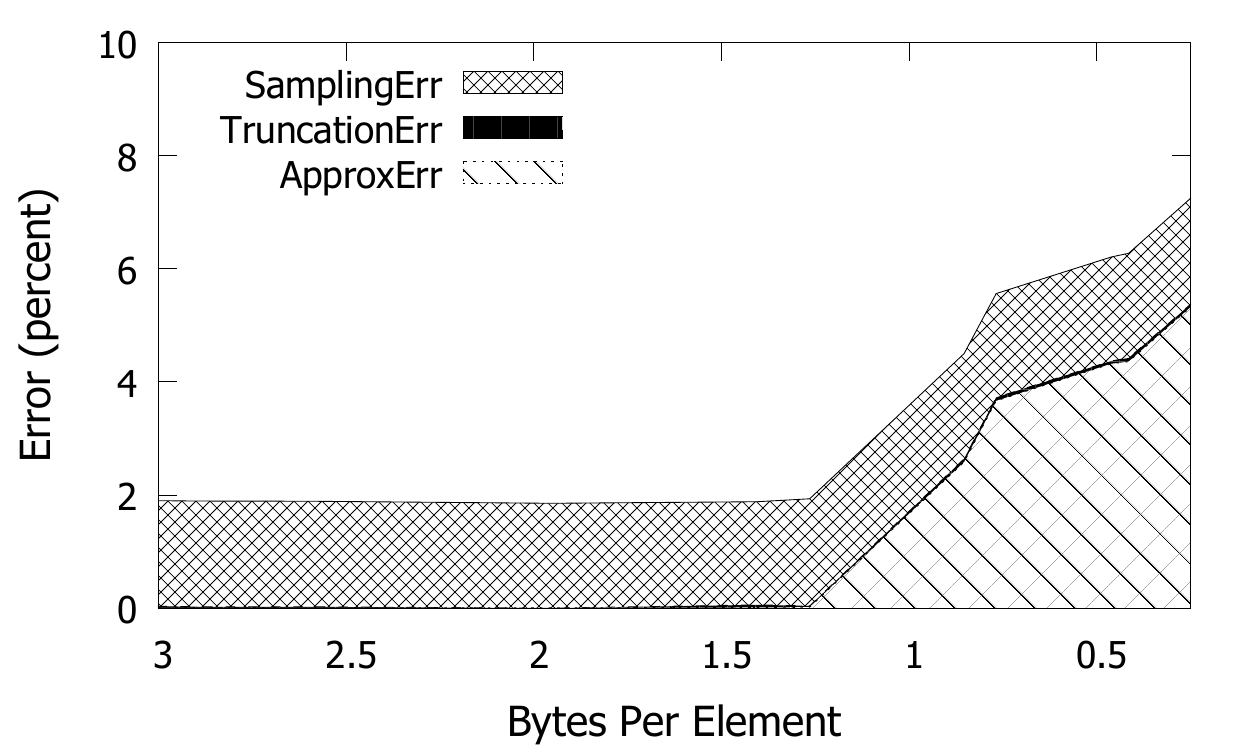}}

}
\caption{Different errors for 1k items cache under Zipf 0.9 distribution, for different sample sizes. }
\label{fig:summery}
\end{figure}

\section{Conclusion}
\label{sec:discussion}
We have introduce TinyLFU, an approximate frequency based cache admission policy.
We showed that TinyLFU can augment caches of arbitrary eviction policy and significantly improve their performance.
We also optimized the memory consumption of TinyLFU using adaptation of known approximate counting techniques with novel techniques tailored specifically
in order to achieve low memory footprint for caches.

Unlike previous schemes that combine ghost cache entries into the eviction policy considerations, in TinyLFU the admission policy is orthogonal to the eviction policy.
This decoupling follows the separation of concerns principle, simplifies design and implementation, and enables optimizing each independently.
In particular, the use of approximate sketching techniques for the admission policy enables maintaining statistics for a relatively large sample while consuming only small amounts of meta-data.


The performance of TinyLFU and its variant W-TinyLFU have been explored and validated through simulations using both synthetic traces as well as multiple real world traces from several different sources.
As mentioned before, TinyLFU is freely available in open source as part of the Caffeine Java caching project~\cite{CaffeineProject}.
Looking into the future, we look into an improved implementation based on succinct hash table design~\cite{TinyTable,TinySet}.

{
\bibliographystyle{plain}
\bibliography{refs}

}

\end{document}